\documentclass[a4paper,11pt]{amsart}

\usepackage[top=2.8cm,bottom=3cm, outer=2cm, inner=2cm]{geometry}
\usepackage{amsfonts,amsmath,amssymb,bbm,amsthm}
\usepackage{enumitem,graphicx}
\usepackage{stmaryrd}
\usepackage{color,pdflscape} 
\usepackage{hyperref} 

\begin{document}
\newcommand{\R}{\mathbb R}
\newcommand{\Z}{\mathbb Z}
\newcommand{\E}{\mathbb E}
\newcommand{\1}{\mathbbm 1}
\newcommand{\N }{\mathbb N}
\newcommand{\q }{\mathbb Q}
\newcommand{\p }{\mathbb P}
\newcommand{\rr }{\mathsf R}
\newcommand{\D }{\mathcal D}
\newcommand{\F }{\mathcal F}
\newcommand{\G }{\mathcal G}
\newcommand{\h }{\mathcal H}
\newcommand{\M }{\mathcal M}
\newcommand{\eps}{\varepsilon}
\newcommand{\argmax}{\operatornamewithlimits{argmax}}
\newcommand{\argmin}{\operatornamewithlimits{argmin}}
\newcommand{\esssup}{\operatornamewithlimits{esssup}}
\newcommand{\essinf}{\operatornamewithlimits{essinf}}

\newcommand{\nc}{\newcommand}
\nc{\bg}{\begin} \nc{\e}{\end} \nc{\bi}{\begin{itemize}} \nc{\ei}{\end{itemize}} \nc{\be}{\begin{enumerate}} \nc{\ee}{\end{enumerate}} \nc{\bc}{\begin{center}} 
\nc{\ec}{\end{center}} \nc{\fn}{\footnote} \nc{\bs}{\backslash} \nc{\ul}{\underline} \nc{\ol}{\overline} 
\nc{\np}{\newpage}  \nc{\fns}{\footnotesize} 
\nc{\scs}{\scriptsize} \nc{\RA}{\Rightarrow} \nc{\ra}{\rightarrow} \nc{\bfig}{\begin{figure}} \nc{\efig}{\end{figure}} \nc{\can}{\citeasnoun} 
\nc{\vp}{\vspace} \nc{\hp}{\hspace}\nc{\LRA}{\Leftrightarrow}\nc{\LA}{\Leftarrow}\nc{\sgn}{\text{sgn}}
\renewcommand{\tilde}{\widetilde}
\nc{\eq}{\end{equation}} 

\nc{\ch}{\chapter}
\nc{\s}{\section}
\nc{\subs}{\subsection}
\nc{\subss}{\subsubsection}

\newtheorem{thm}{Theorem}[section]
\newtheorem{cor}[thm]{Corollary}
\newtheorem{lem}[thm]{Lemma}

\theoremstyle{remark}
\newtheorem{rem}[thm]{Remark}

\theoremstyle{example}
\newtheorem{ex}[thm]{Example}

\theoremstyle{ass}
\newtheorem{ass}[thm]{Assumption}

\theoremstyle{definition}
\newtheorem{df}[thm]{Definition}

\newenvironment{rcases}{
  \left.\renewcommand*\lbrace.
  \begin{cases}}
{\end{cases}\right\rbrace}

\setlength\parindent{0pt}
\pagestyle{plain}

\title{A weak law of large numbers for a limit order book model\\ with fully state dependent order dynamics}


   \author{Ulrich Horst}
   \address{Humboldt-Universit\"at zu Berlin, Germany}
   \email{horst@math.hu-berlin.de}

   \author{D\"orte Kreher}
   \address{Humboldt-Universit\"at zu Berlin, Germany}
   \email{kreher@math.hu-berlin.de}

	\begin{abstract}
This paper studies a limit order book (LOB) model,
in which the order dynamics depend on both, the current best available prices and
the current volume density functions. For the joint dynamics of the best
bid price, the best ask price, and the standing volume densities on both sides of the LOB we derive a weak law of large numbers, which states that the LOB model converges to a
continuous-time limit when the size of an individual order as well as the
tick size tend to zero and the order arrival rate tends to infinity. In
the scaling limit the two volume densities follow each a non-linear PDE
coupled with two non-linear ODEs that describe the best bid and ask price. 
	\end{abstract}
	

  \thanks{This research was partially supported by CRC 649: Economic Risk. We thank Moritz Greving for his assistance with the simulation results.}

\maketitle

\renewcommand{\baselinestretch}{1.15}\normalsize
\setlength{\parskip}{5pt}

\s{Introduction}

While limit order books have extensively been discussed in the economic and econometric literature for some years (cf.~for example \cite{Biais,Easley,Glosten,Rosu}), they have only recently gained increased attention by researchers in mathematical finance. One research objective is to specify a realistic discrete dynamics of a LOB which can be approximated by an analytically tractable continuous time model. This is achieved by introducing scaling parameters and passing to the high frequency limit, when the number of submitted orders gets large while the individual order size and the tick size tend to zero. 
Depending on the scaling assumptions the high frequency limit will either be deterministic as in a law of large numbers or be of (jump) diffusion type as in a functional central limit theorem. 

Deterministic high frequency limits for LOB models were derived by \cite{HP} and \cite{Gao}. In \cite{HP} a weak law of large numbers is established for a limit order book model with Markovian dynamics depending on prices only. 
In \cite{Gao} the authors study a limit order book model, similar to ours but without any feedback effect, and derive a deterministic ODE limit using weak convergence in the space of positive measures on a compact interval. 
A diffusion limit for order book dynamics can be found in \cite{Cont1}, where the top of the book is analyzed. The result was later generalized in \cite{Cont2}. In \cite{Lakner1} a high frequency limit for a one-sided limit order book model is derived under the assumption that on average investors place their limit orders above the current best ask price. The opposite case when orders are placed in the spread with higher probability is analyzed in \cite{Lakner2}, where the authors use a coupling between a simple one-sided limit order book model and a branching random walk to characterize the diffusion limit, cf.~also \cite{Simatos}. In the recent preprint \cite{Guo} the limiting behaviour of an individual order position together with the best bid and best ask queue is studied and fluctuations around their fluid limits are derived.

There is considerable empirical evidence (see, e.g. \cite{Biais, Cebiroglu-Horst, Hautsch-Huang} and references therein) that the state of the order book, especially order imbalance at the top of the book, has a noticeable impact on order dynamics. However, in the literature the order flow in most limit order book models either follows independent Poisson dynamics or depends on the price process only as in \cite{HP}. Exceptions to this are \cite{Abergel2}, where Hawkes-type dynamics are used, as well as \cite{Rosenbaum1} and the very recent preprint \cite{Rosenbaum2}, in which the ergodicity of a general Markovian order book model is studied and the diffusivity of the rescaled price process in this general framework is derived.

In this paper we adapt the model from \cite{HP} but use a different approach which allows us to deal with much more general, in fact fully state dependent Markovian order flow dynamics:
the type of order (market order, limit order placement, cancellation), its size, and the price level at which the order is submitted can all depend on the current state of the limit order book, i.e.~on prices as well as on standing volumes. This is different from \cite{BHQ}, where the standing volume only influences the price dynamics, but there is no direct feedback to the order flow. As a result, unlike in \cite{HP}, the price process cannot be analyzed separately. Instead, we have to establish joint convergence of prices and volumes. The resulting scaling limit for a fully state dependent Markovian order book dynamics is the main result of this paper. 

Our main theorem states that when the number of submitted orders goes to infinity over a fixed time horizon, while the proportion of active orders, the tick size, and the individual order size tend to zero, the dynamics of the prices and the volume density functions converge to the unique solution of a non-linear coupled ODE/PDE system. To prove our main result we first construct a deterministic discrete non-linear approximation $\tilde{S}^{(n)}$ to the random discrete order book dynamics $S^{(n)}$. This is done using a weak law of large numbers for triangular martingale difference arrays as in \cite{HP}, even though our method of approximation is different and more elegant, which allows us to handle this more general setting. In the next step we then construct an iteration towards the deterministic approximation for fixed $n$, denoted $\tilde{S}^{(n),m}$, and we prove that it approximates $\tilde{S}^{(n)}$ almost uniformly. Afterwards it is shown that each iteration step in the prelimit converges as $n$ goes to infinity to a continuous model $\hat{S}^m$ solving a certain differential equation. Indeed, these models can be seen to be a fixed point iteration generated by a contraction mapping. The fixed point then gives a solution to our limiting coupled ODE/PDE system. 

For the ease of notation we have chosen to analyze only the buy side of the order book together with the bid price in most parts of this paper. However, if one defines the sell side and ask price in an analogous way, the result can easily be extended to a two-sided order book with order dynamics depending on the {\it whole} limit order book, i.e.~on bid and ask prices as well as the order volumes of both sides of the book. The corresponding result for the two-sided LOB model can be found in the final section of the paper. Especially, making the distribution of order types depending on the bid-ask spread will ensure that the bid and ask price do not cross, cf.~also \cite{HP}. Moreover, we assume that order arrival times are deterministic. However, one can easily generalize our main result allowing for randomly spaced arrival times by making use of the time change theorem as has been done in \cite{BHQ,HP}.

The remainder of this paper is organized as follows: In Section \ref{Setup} we define the dynamics of a sequence of one-sided discrete limit order book models, state our assumptions and the main result. We also give an example which satisfies all our assumptions. Section \ref{limitmodel} is devoted to the analysis of the limiting coupled PDE/ODE system, while Section \ref{proof} contains the convergence proof of the discrete order book models to the high frequency limit. Finally, in Section \ref{2} we state our main theorem for the two-sided limit order book model and conduct a simulation study, which shows how the state dependency can be used to equilibrate the buy and sell side volumes from an initially highly imbalanced volume distribution.   

\s{Setup and main result}\label{Setup}

In this section we define for every $n\in\N$ a model for the dynamics of a one-sided limit order book with tick size $\Delta x^{(n)}$. Later we consider the scaling limit of these models when the tick size and the impact of a single order tend to zero, while the number of order placements and order cancelations over a given time horizon $[0,T]$ tends to infinity. Our modeling framework closely follows \cite{HP} but we allow for a much more general dependence of order arrivals on the state of the book. Throughout, all random variables are defined on a common complete probability space $(\Omega,\F,\p)$.

\subs{The model}

The dynamics of the order book in the $n$-th model is described by a c\`adl\`ag stochastic process $S^{(n)}=\left(S^{(n)}(t)\right)_{0\leq t\leq T}$ taking values in the Hilbert space
\[E:=\R\times L^2(\mathbb{R}),\qquad \left\Vert \alpha\right\Vert_E:=\left|\alpha_1\right|+\left\Vert \alpha_2\right\Vert_{L^2}.\]
The state of the book changes due to arriving market and limit orders. In the $n$-th model there are $\left\lfloor T/\Delta t^{(n)}\right\rfloor$ such {\sl events} taking place at times
\[t_k^{(n)}:=k\Delta t^{(n)},\quad k=1,\dots,\left\lfloor\frac{T}{\Delta t^{(n)}}\right\rfloor,\]
where $\Delta t^{(n)}$ denotes a scaling parameter converging to zero as $n\ra\infty$ and $t_0^{(n)}=0$. 

The state of the book after $k$ events is denoted $S^{(n)}_k$. We put
\[S^{(n)}(t):=S^{(n)}_k:=\left(B_k^{(n)},v_k^{(n)}\right)\quad\text{for} \quad t\in\left[t_k^{(n)},t_{k+1}^{(n)}\right),\]
where $B^{(n)}_k$ and $v_k^{(n)}$ denote the best bid price and the buy side volume density function {\it relative} to the best bid price ({\sl relative volume density function}), respectively. To be precise, defining $x^{(n)}_j:=j\Delta x^{(n)}$ for $j\in\Z,\ n\in\N$,
\[\int_{x_{-j-1}^{(n)}}^{x_{-j}^{(n)}}v_{k}^{(n)}(x)dx\] 
represents the liquidity available for buying at a price which is $j\in\N_0$ ticks below the best bid price at time $t^{(n)}_k$. 
In order to model placements of limit orders inside the spread, the function $v^{(n)}_k,\ k\in\N,$ will be defined on the whole real line. We refer to the buy volumes standing at positive distance from the best bid price as the {\sl shadow book}, cf.~Figure \ref{shadow}. 

\begin{figure}[h]
\includegraphics[width=0.55\textwidth]{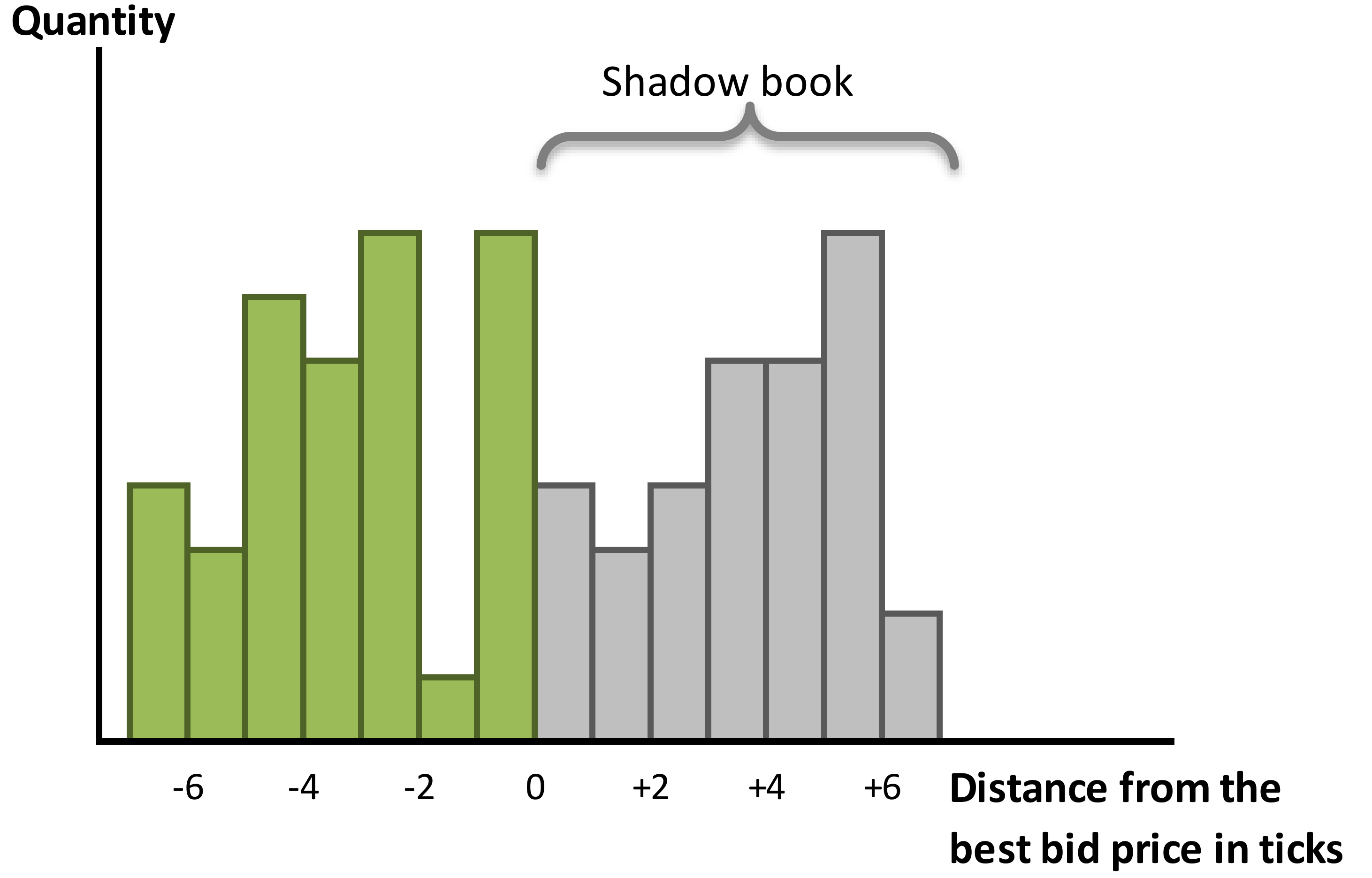}
\caption{The shadow book extends the relative volume density function to the right}
\label{shadow}
\end{figure}

The idea of the shadow book is taken from \cite{HP}. The shadow book has to be understood as a tool to model the (conditional) distribution of the size of limit order placements inside the spread in such a way, that those placements extend the current volume density function of the {\sl visible book} in a sufficiently ``smooth'' way to the right. The shadow book follows the same dynamics as the volumes of the visible book and becomes part of the visible book through price changes. The working of the shadow book and its interaction with the visible book will be further explained below; cf. Example \ref{ex-shadow-book}.

At time $t=0$ the state of the limit order book is deterministic for all $n\in\N$ and denoted by 
\[s_0^{(n)}=\left(B_0^{(n)},v^{(n)}_{0}\right) \in \mathbb{R} \times L^2(\R).\]
To state the convergence condition on the sequence of initial states we introduce for each $n\in\N$ the translation operators $T_+^{(n)}$ and $T_-^{(n)}$, which act on functions $f:\R\ra\R$ in the following way:
\[T_+^{(n)}(f)(\cdot):=f\left(\cdot+\Delta x^{(n)}\right),\qquad T_-^{(n)}(f)(\cdot):=f\left(\cdot-\Delta x^{(n)}\right).\]
Note that the translation operator is isometric, i.e. for all $f\in L^2$,
\[\left\Vert T^{(n)}_+(f)\right\Vert_{L^2}=\left\Vert f \right\Vert_{L^2}.\]

Furthermore, let us fix some constant $M>0$ throughout.

\begin{ass}\label{initial} 
The initial volume function $v^{(n)}_{0}$ is a non-negative step-function on the grid \mbox{$\{x_j^{(n)},\ j\in\Z\}$,} which is uniformly bounded by $M$ and has compact support in $[-M,M]$ for all $n\in\N$. Moreover, 
there exists a non-negative continuously differentiable function $v_{0}\in L^2$ such that
\[\left\Vert v_{0}^{(n)}-v_{0}\right\Vert_{ L^2}=\mathcal{O}\left(\Delta x^{(n)}\right).\]
Also there exists $B_0\in\R_+$ such that $B^{(n)}_0\ra B_0.$ We denote $s_0:=\left(B_0,v_{0}\right)\in E$.
\end{ass}

\begin{rem}\label{initialshift}
Note that Assumption \ref{initial} implies that $v_0$ also has compact support in $[-M,M]$ and therefore, because $v_0\in C^1(\R)$, its derivative must be bounded by some $C>0$. Thus,
\begin{eqnarray*}
\left\Vert\left(T_{+}^{(n)}-I\right)\left(v_{0}^{(n)}\right)\right\Vert_{L^2}
&\leq& \left\Vert \left(T_{+}^{(n)}\right)\left(v_0^{(n)}-v_{0}\right)\right\Vert_{L^2}+\left\Vert\left(T_{+}^{(n)}-I\right)\left(v_{0}\right)\right\Vert_{L^2}+\left\Vert v_0-v_{0}^{(n)}\right\Vert_{L^2}\\
&\leq&\mathcal{O}\left(\Delta x^{(n)}\right)+\sup_{x\in\R}\left|v_0'(x)\right|\left\Vert\1_{[-M,M]}\Delta x^{(n)}\right\Vert_{L^2}=\mathcal{O}\left(\Delta x^{(n)}\right).
\end{eqnarray*}
\end{rem}

In our model there are three events that change the state of our limit order book. The (buy side) limit order book changes if:
\bi
\item (A): a market sell order of size equal to the current best bid queue arrives. In this case the best bid price decreases by one tick. Hence, the relative volume density function shifts one tick to the right.
\item (B): a buy limit order is placed inside the spread one tick above the current best bid price. In this case the best bid price increases by one tick and the relative volume density function shifts one tick to the left.
\item (C): a buy limit order placement of size $\frac{\Delta v^{(n)}}{\Delta x^{(n)}}\omega_k^{(n)}$ at price level $\rho_k^{(n)}$ occurs. If $\omega_k^{(n)}<0$, this corresponds to a cancelation of volume.
\ei

Here $\Delta v^{(n)}$ is a scaling parameter that determines the size of an individual placement / cancelation. We refer to market orders and limit buy order placements in the spread (Types A,B) as {\sl active orders}. They lead to price changes. Cancelations and limit order placements (Type C) do not lead to price changes. They are referred to as {\sl passive orders}. The assumption that market orders match precisely against the volume standing at the top of the book and hence shift prices by exactly one tick is made for convenience and shows that our framework is flexible enough to allow for larger market orders. However, it is not unrealisitc: in an empirical study the authors of \cite{Farmer} found that in their data sample around 85\% of the sell market orders which lead to price changes match exactly the size of the volume standing at the best bid price. 
The effect of a market order that does not lead to a price change is equivalent to a cancelation of standing volume. 

Event types are determined by a field of random variables $\left(\phi_k^{(n)}\right)_{k,n\in\N}$ taking values in the set $\left\{A,B,C\right\}$. The size and the price level at which an order placement resp.~cancelation takes place are determined by a field of random variables $\left(\omega_k^{(n)},\rho_k^{(n)}\right)_{k,n\in\N_0}$ according to the following assumption. 

\begin{ass}\label{density}
There exists a field of random variables $\left(\pi_k^{(n)}\right)_{k,n\in\N_0}$ taking values in the compact interval $[-M,M]$ almost surely and
\[\rho_k^{(n)}:=B_k^{(n)}+j\Delta x^{(n)}\quad\text{for}\quad\pi_k^{(n)}\in\left[x_{j-1}^{(n)},x_{j}^{(n)}\right).\]
Furthermore, there exists a field of random variables $\left(\omega_k^{(n)}\right)_{k,n\in\N_0}$ such that $\omega_k^{(n)}\in[-M,M]$ for all $k,n\in\N_0$.
\end{ass}

The random variables $\pi_k^{(n)},\ k,n\in\N_0,$ determine the placement/cancelation price levels relative to the best bid price. If $\rho_k^{(n)}= B_k^{(n)}$, then the placement/cancelation takes place at the best bid price; if $\rho_k^{(n)}< B_k^{(n)}$, then it takes place deeper in the book; else it takes place in the shadow book. The shadow book interacts with the visible book through price changes which shift the relative volume density functions $v^{(n)}$. The following example illustrates the working of the shadow book. 

\begin{ex}\label{ex-shadow-book}
Suppose that the $k$th event is a limit order placement one tick above the best bid price into the shadow book, i.e.
\[
	\phi_k^{(n)} = C,\quad \rho_k^{(n)}=B_k^{(n)}+\Delta x^{(n)} \quad \mbox{and} \quad \omega_k^{(n)}>0.
\] 
Further suppose that the $(k+1)$st event is a buy limit order placement in the spread, i.e.~\mbox{$\phi_{k+1}^{(n)}=B$.} Then,
\[
	B^{(n)}_{k+2} = B^{(n)}_{k+1} + \Delta x^{(n)} = B^{(n)}_{k} + \Delta x^{(n)}
\] 
and for all $x\in\left[-\Delta x^{(n)},0\right)$ corresponding to standing volumes at the current best bid price,
\[
	v^{(n)}_{k+2}(x) = v^{(n)}_{k+1}\left(x+\Delta x^{(n)}\right) = v^{(n)}_{k} \left(x+\Delta x^{(n)}\right)+ \frac{\Delta v^{(n)}}{\Delta x^{(n)}}\omega^{(n)}_k,
\]
while for all $x\notin\left[-\Delta x^{(n)},0\right)$,
\[v^{(n)}_{k+2}(x) = v^{(n)}_{k+1}\left(x+\Delta x^{(n)}\right) = v^{(n)}_{k} \left(x+\Delta x^{(n)}\right).\]
\end{ex}

Note that in general $\omega_k^{(n)}$ is (even conditionally) dependent on $\pi_k^{(n)}$, if one wants to avoid negatives volumes due to cancelations. 
The main contribution of this paper is that the conditional distribution of the random variables $\left(\phi_k^{(n)},\omega_k^{(n)}, \pi_k^{(n)}\right)_{k,n\in\N_0}$ may depend on both, current prices and volumes. This extends \cite{HP} where prices are independent of volumes as well as \cite{BHQ} where only the distributions of price increments depend on volumes. 

 For every $n\in\N_0$ and $k=0,1,\dots,\lfloor T/\Delta t^{(n)}\rfloor$ we define the $\sigma$-field $\F_k^{(n)}:=\sigma\left(S_j^{(n)},\ j\leq k\right)$.
We will assume that for each $n\in\N$ the state process $S^{(n)}$ is a Markov process in its own filtration, cf.~Assumption \ref{Lipschitz} below. 
 
To formulate the next assumption we need to introduce the space $E':=\{s=(B,v)\in E:\ v\in C^1\}.$

\begin{ass}\label{Lipschitz}$\quad$
\begin{enumerate}
\item There are two Lipschitz continuous functions $p^A,p^B:E\ra[0,1]$ with Lipschitz constant $L$ and a scaling parameter $\Delta p^{(n)}$ such that for all $n\in\N_0$ and $k\leq \lfloor T/\Delta t^{(n)}\rfloor$,
\begin{eqnarray*}
\p\left(\left.\phi^{(n)}_k=I \ \right|\F_k^{(n)}\right)=\Delta p^{(n)} p^I\left[S_k^{(n)}\right] \quad a.s.\quad\text{for}\ I=A,B.
\end{eqnarray*}
\item There are Lipschitz continuous functions $f^{(n)}:E\ra L^2,\ n\in\N_0,$ with common Lipschitz constant $L>0$ such that for all $k\leq \lfloor T/\Delta t^{(n)}\rfloor$,
\[
	f^{(n)}\left[S_k^{(n)}\right](\cdot)=\frac{1}{\Delta x^{(n)}}\E\left(\left. \omega_k^{(n)}\sum_{j\in\Z} \1_{\left\{\pi_k^{(n)}\in\left[x_j^{(n)},x^{(n)}_{j+1}\right)\right\}}(\cdot)\1_C\left(\phi_k^{(n)}\right)\right|\F_k^{(n)}\right)\quad a.s.
	\]
and
\[\sup_{s\in E}\left\Vert f^{(n)}[s](\cdot)\right\Vert_\infty\leq M.\]
Moreover, there exists a function $f:E\ra L^2$ such that
\[\sup_{s\in E}\left\Vert f^{(n)}[s]-f[s]\right\Vert_{L^2}=\mathcal{O}\left(\Delta x^{(n)}\right),\]
where $f[s](\cdot):\R\ra[-M,M]$ is continuously differentiable in $x$ for all $s\in E'$ with derivate being uniformly bounded in absolute value by $M$.
\end{enumerate}
\end{ass}

\begin{rem}\label{shiftf}
Note that Assumption \ref{Lipschitz} implies that
\[\sup_{s\in E}\left\Vert\left(T_+^{(n)}-I\right)\left(f^{(n)}[s]\right)\right\Vert_{L^2}=\mathcal{O}\left(\Delta x^{(n)}\right).\]
\end{rem}

The following example illustrates how our modeling framework allows for a dependence on the price dynamics and standing volumes. 

\begin{ex}\label{Ex1}
Given any $h\in L^2$ define 
\[H^{(n)}_{k}:=\int_{\R_+}v^{(n)}_{k}(x)h(x)dx.\]
We may interpret $H^{(n)}_k$ as an indicator for the volume standing at the top of or deeper into the book, depending on the choice of $h$.   
Set
\[
	p^I\left[S_k^{(n)}\right]:=g^I\left(B_k^{(n)},H^{(n)}_{k}\right), \quad I=A,B,
\]
for Lipschitz continuous functions $g^I:\R^2\ra[0,1]$. Then by the Cauchy-Schwarz inequality there exists $L<\infty$ such that 
\begin{eqnarray*}
&&\left|p^{I}\left[S_k^{(n)}\right]-p^{I}\left[\tilde{S}_k^{(n)}\right]\right|\leq L \left\Vert S_k^{(n)}-\tilde{S}_k^{(n)}\right\Vert_E.
\end{eqnarray*}
Let us further assume that for all $n\in\N_0,\ k\leq \lfloor T/\Delta t^{(n)}\rfloor$, $I=A,B$,
\begin{eqnarray*}
	\p\left(\left.\phi^{(n)}_k=I\ \right|S_j^{(n)},\ j\leq k\right) & = & \Delta p^{(n)} p^I\left[S_k^{(n)}\right]\quad a.s.
\end{eqnarray*}
and that there exists a function $m:\R^2\ra L^2$ such that for all $k,n$ as above
\[
	\p\left(\left.\pi_k^{(n)}\in dx\ \right|S_j^{(n)},\ j\leq k\right)=m\left[ B_k^{(n)},H^{(n)}_{k}\right](x)dx,
\]
where for each $(b,h)\in\R^2$ the function $m[b,h](\cdot)$ is uniformly bounded with bounded support in $[-M,M]$ and the mapping $(b,h,x) \mapsto m[ b,h] (x)$ is continuously differentiable with bounded derivatives. 
Moreover, suppose that there is a field $\left(\tilde{\omega}_k^{(n)}\right)$ of i.i.d.~random variables with bounded density which has compact support in $[-M,M]$ and set
\[\omega_k^{(n)}=\tilde{\omega}_k^{(n)}\wedge\left(- v^{(n)}\left(\pi_k^{(n)}\right)+\eps\right)\qquad\text{for some }\eps>0.\]
If the random variables $\phi^{(n)}_k, \pi^{(n)}_k,$ and $\tilde{\omega}_k^{(n)}$ are conditionally on $\{S_j^{(n)}, j\leq k\}$ independent and if $\tilde{\omega}_k^{(n)}$ is also independent of $\{S_j^{(n)}, j\leq k\}$, then Assumption \ref{Lipschitz} is satisfied.
\end{ex}

\subs{Main result}

We are now ready to define the full dynamics of the order book. For notational convenience we define for $I\in\left\{A,B,C\right\}$ and $k,n\in\N$, the event indicator variable
\[
	\1^{(n),I}_k:=\1_I\left(\phi_k^{(n)}\right)
\]
and introduce the short-hand notation ($I=A,B$):
\[p^{(n),I}[\cdot]:=\Delta p^{(n)}p^I[\cdot], \quad p^{(n),B-A}:=p^{(n),B}-p^{(n),A},\quad p^{B-A}:=p^B-p^A,
\quad \1^{(n),B-A}_k:=\1^{(n),B}_k-\1^{(n),A}_k.
\]

\begin{df}
For each $n\in\N$ the dynamics of the state process $S^{(n)}=\left(B^{(n)},v^{(n)}\right)$ is given by $S_0^{(n)}:=s_0^{(n)}$ and for $k=1,\dots,\left\lfloor\frac{T}{\Delta t^{(n)}}\right\rfloor$,
\begin{equation}
\begin{split}
	B_{k}^{(n)} &= B^{(n)}_{k-1}+\Delta x^{(n)}\1_{k-1}^{(n),B-A} \\
	v_{k}^{(n)} &= v_{k-1}^{(n)}+\left(T_-^{(n)}-I\right)\left(v_{k-1}^{(n)}\right)\1_{k-1}^{(n),A}+\left(T_+^{(n)}-I\right)\left(v_{k-1}^{(n)}\right)\1_{k-1}^{(n),B} +\Delta v^{(n)} M_{k-1}^{(n)},
\end{split}
\end{equation}
where 
\[M_{k}^{(n)}(\cdot):=\1_{k}^{(n),C}\frac{\omega_k^{(n)}}{\Delta x^{(n)}}\sum_{j\in\Z} \1_{\left\{\pi_k^{(n)}\in\left[x_j^{(n)},x^{(n)}_{j+1}\right)\right\}}(\cdot).\]
\end{df}

To derive a law of large numbers we need to make the right assumptions on the scaling parameters. Our choice of scaling introduces two time scales, a fast one for limit order arrivals and cancelations and a comparably slow one for market order arrivals and limit order placements in the spread. 

\begin{ass} \label{scaling} 
There exist constants $c_0,c_1,c_2>0$ and $\beta\in(0,1)$ such that 
\[\lim_{n\ra\infty}\frac{\Delta x^{(n)}\Delta p^{(n)}}{\Delta t^{(n)}}=c_0,\quad\lim_{n\ra\infty}\frac{\Delta v^{(n)}}{\Delta t^{(n)}}=c_1,\quad \lim_{n\ra\infty}\frac{\Delta x^{(n)}}{\left(\Delta t^{(n)}\right)^\beta}=c_2.\]
W.l.o.g.~we assume that $c_0=c_1=c_2=1$ in the following.
\end{ass}

\begin{rem}
While it is very natural to assume that $\Delta v^{(n)}\sim \Delta t^{(n)}$ for $n\ra\infty$ in order to keep the total volume of orders in the limit order book of constant size, the assumption $\Delta x^{(n)}\Delta p^{(n)}\sim \Delta t^{(n)}$ is not so standard. However, note that this constitutes the critical (and interesting) case. Indeed, as can be easily seen from the proof of our main theorem, assuming that $\Delta x^{(n)}\Delta p^{(n)}=o(\Delta t^{(n)})$ would lead to a constant price in the high frequency limit. Such a result can be found in \cite{Gao}.
\end{rem}

The following weak law of large numbers is the main result of this paper. It states that the state process converges in probability to a deterministic limit that can be described as the solution of a system of non-linear differential equations subject to an initial boundary condition.

\begin{thm}\label{limit}
Under Assumptions \ref{initial}, \ref{density}, \ref{Lipschitz}, and \ref{scaling} there exists a deterministic process $S:[0,T]\ra E$ such that for all $\eps>0$,
\[\lim_{n\ra\infty}\p\left(\sup_{0\leq t\leq T}\left\Vert S^{(n)}(t)-S(t)\right\Vert_E>\eps\right)=0. \]
The function $S=(B,v)$ is the unique classical solution to the following coupled ODE/PDE initial boundary value problem:
\begin{equation}\label{system}
\begin{aligned}
S(0)&=s_0,\\
dB(t)&=p^{B-A}[S(t)]dt,\quad t\in[0,T],\\
v_t(t,x)&=p^{B-A}[S(t)]v_x(t,x)+f[S(t)](x),\quad (t,x)\in[0,T]\times\R.
\end{aligned}
\end{equation}
\end{thm}

The rest of the paper is devoted to the proof of Theorem \ref{limit}. In the next section, we first show with the help of a fixed point argument that the ODE/PDE system (\ref{system}) does indeed have a unique solution. Section \ref{proof} contains the proof of convergence of the discrete models to the scaling limit. The difficulty in proving Theorem \ref{limit} comes from the non-local dependence of the coefficients on the whole function in (\ref{system}).

\begin{ex}
Let us choose $\phi_k^{(n)},\pi_k^{(n)},$ and $\tilde{\omega}_k^{(n)}$ as in Example \ref{Ex1} and set for some $\tilde{M}>0$,
\[\omega_k^{(n)}:=\left(\tilde{\omega}_k^{(n)}\right)^+-\frac{1}{\tilde{M}+1}\left(\tilde{\omega}_k^{(n)}\right)^-\left[v_k^{(n)}\left(\pi_k^{(n)}\right)\wedge \tilde{M}\right].\]
Then, choosing $\tilde{M}$ large enough, the ODE/PDE system (\ref{system}) takes the special form 
\begin{equation*}
\begin{aligned}
S(0)&=s_0,\\
dB(t)&=p^{B-A}[B(t),H(t)]dt,\quad t\in[0,T],\\
v_t(t,x)&=p^{B-A}[B(t),H(t)]v_x(t,x)+f_1[B(t),H(t)](x)-f_2[B(t),H(t)](x)v(t,x),\quad (t,x)\in[0,T]\times\R.
\end{aligned}
\end{equation*}
\end{ex}

Throughout the paper we will denote by $C>0$ a generic constant that may vary from line to line and is independent of any index involved.

\s{The limit model}\label{limitmodel}

In this section we prove existence and uniqueness of a solution to (\ref{system}). First, we explicitly construct a solution via a fixed point iteration on a suitable Banach space to prove existence. Under the assumptions of Theorem \ref{limit} there exists a constant $\ol{K}<\infty$ such that 
\[\left\Vert \E S^{(n)}_{k}\right\Vert_E\leq B_0^{(n)}+ T\cdot\frac{\Delta x^{(n)}\Delta p^{(n)}}{\Delta t^{(n)}}+\left\Vert v_{0}^{(n)}\right\Vert_{L^2}+T\cdot\frac{\Delta v^{(n)}}{\Delta t^{(n)}}\cdot\sup_{s\in E}\left\Vert f^{(n)}[s]\right\Vert_{L^2}\leq \ol{K}\]
for all $k,n\in\N$. We thus choose as our Banach space the space $\tilde{E}$ of functions $g:\ [0,T]\ra E$ which satisfy
\[\sup_{t\in[0,T]}\left\Vert g(t)\right\Vert_E\leq \ol{K}\]
equipped with the norm $\sup_{t\in[0,T]}\left\Vert g(t)\right\Vert_E$. Uniqueness will be shown using a standard Gronwall argument.

\subs{Fixed point iteration in the scaling limit}\label{fixc}

To construct a solution to (\ref{system}) we perform a fixed point iteration for the function $F: \tilde{E}\ra\tilde{E}$ defined through $F: g\mapsto G$, where $G=(G^1,G^2)$ and $G^1: [0,T]\ra\R$ and $G^2: [0,T]\ra L^2$ are given by
\begin{eqnarray*}
G^1(t) &=& B_0+\int_0^tp^{B-A}\left[g(s)\right]ds,\\
G^2(t,x) &=& v_{0}\left(x+\int_0^tp^{B-A}\left[g(s)\right]ds\right)+\int_0^t f\left[g(s)\right]\left(x+\int_s^tp^{B-A}\left[g(u)\right]du\right)ds.
\end{eqnarray*}

Note that by definition and Assumptions \ref{initial}, \ref{density}, and \ref{Lipschitz}, $G^2(t,\cdot)$ has support in the compact interval $[-M-T,M+T]$ for all $t\in[0,T]$.

We define $\tilde{E}':=\{g\in\tilde{E}\ |\ g:[0,T]\ra E'\}$.

\begin{lem}\label{dif}
For fixed $g\in\tilde{E}'$ the function $G=(G^1,G^2)$ satisfies
\begin{eqnarray*}
G(0)&=&s_0,\\
dG^1(t)&=&p^{B-A}\left[g(t)\right]dt,\\
G^2_t(t,x)&=&p^{B-A}[g(t)]G^2_x(t,x)+f[g(t)](x)\quad \forall\ (t,x)\in[0,T]\times\R.
\end{eqnarray*}
Moreover, in this case $G \in\tilde{E}'$ and there exist two constants $\ol{J},\ol{L}<\infty$, which do not depend on $g$, such that
\[\left|G^2_x(t,x)\right|\leq\ol{J},\quad \left|G^2_t(t,x)\right|\leq\ol{L}\quad\forall\ (t,x)\in[0,T]\times\R.\]
\end{lem}

\begin{proof}
It follows from the general theory of first-order PDEs that $G$ solves the claimed PDE. Moreover, by Assumptions \ref{initial} and \ref{Lipschitz},
\begin{eqnarray*}
\left|G^2_x(t,x)\right|&\leq&\sup_{x\in\R}\left|v_{0}'(x)\right|+\ T\cdot\sup_{\substack{x\in\R\\ s\in E'}}\left|\left(f[s]\right)'(x)\right|=:\ol{J}<\infty,\\
\left|G^2_t(t,x)\right|&\leq&\left|G^2_x(t,x)\right|+\left|f[g(t)](x)\right|\leq\ol{J}+M=:\ol{L}<\infty
\end{eqnarray*}
as well as
\[\sup_{t\in[0,T]}\left\Vert G(t)\right\Vert_E\leq B_0+T+\left\Vert v_{0}\right\Vert_{L^2}+T\cdot \sup_{s\in E}\left\Vert f[s]\right\Vert_{L^2}\leq \ol{K}.\]
\end{proof}

If we find a fixed point of $F$ which lies in $\tilde{E}'$, then Lemma \ref{dif} tells us that it must indeed solve (\ref{system}). To do this, we will show that the function $F$ is Lipschitz continuous to conclude with Banach's fixed point theorem. 
In the following we write $F_{B,t}(g):=G^1(t)$ and $F_{v,t}(g):=G^2(t)$ for $G$ defined as above. Then:
\[\left|F_{B,t}(g)-F_{B,t}(\tilde{g})\right|\leq\int_0^t\left|p^{B-A}[g(s)]-p^{B-A}\left[\tilde{g}(s)\right]\right|ds\leq 2L\int_0^t\left\Vert g(s)-\tilde{g}(s)\right\Vert_Eds.\]
Moreover,
\begin{flalign*}
\left\Vert F_{v,t}(g)-F_{v,t}(\tilde{g})\right\Vert_{L^2} & \leq\left\Vert v_{0}\left(\cdot+\int_0^tp^{B-A}\left[g(s)\right]ds\right)-v_{0}\left(\cdot+\int_0^tp^{B-A}\left[\tilde{g}(s)\right]ds\right)\right\Vert_{L^2}\\
&~~+\int_0^t \left\Vert f\left[g(s)\right]\left(\cdot+\int_s^tp^{B-A}\left[g(u)\right]du\right)-f\left[\tilde{g}(s)\right]\left(\cdot+\int_s^tp^{B-A}\left[\tilde{g}(u)\right]du\right)\right\Vert_{L^2} ds.
\end{flalign*}
Since $v_0'$ is uniformly bounded with bounded support the mean value theorem along with Assumption \ref{Lipschitz} yields,
\begin{eqnarray*}
\left\Vert v_{0}\left(\cdot+\int_0^tp^{B-A}\left[g(s)\right]ds\right)-v_{0}\left(\cdot+\int_0^tp^{B-A}\left[\tilde{g}(s)\right]ds\right)\right\Vert_{L^2}\leq C \int_0^t\left\Vert g(s)-\tilde{g}(s)\right\Vert_Eds.
\end{eqnarray*}
Similarly, as $\left|\left(f[s]\right)'(\cdot)\right|$ is uniformly bounded and has bounded support in $[-M,M]$ for all $s\in E'$ and as $f$ is Lipschitz continuous,
\begin{eqnarray*}
&&\left\Vert f\left[g(s)\right]\left(\cdot+\int_s^tp^{B-A}\left[g(u)\right]du\right)-f\left[\tilde{g}(s)\right]\left(\cdot+\int_s^tp^{B-A}\left[\tilde{g}(u)\right]du\right)\right\Vert_{L^2}\\
&\leq& C\int_s^t\left|p^{B-A}[g(u)]-p^{B-A}\left[\tilde{g}(u)\right]\right|du+\left\Vert f\left[g(s)\right]-f\left[\tilde{g}(s)\right]\right\Vert_{L^2}\\
&\leq& 2CL \int_s^t\left\Vert g(u)-\tilde{g}(u)\right\Vert_Edu+L\left\Vert g(s)-\tilde{g}(s)\right\Vert_E.
\end{eqnarray*}
Hence, there exists some $\hat{K}>0$ such that for all $t\in[0,T]$,
\[\left\Vert F_{t}(g)-F_{t}(\tilde{g})\right\Vert_{E}\leq\hat{K}\int_0^t\left\Vert g(s)-\tilde{g}(s)\right\Vert_Eds.\]
Now, the space $\tilde E$ is also a Banach space with respect to the equivalent weighted norm
\[\left\Vert g\right\Vert_*:=\sup_{0\leq t\leq T}e^{-\alpha t}\left\Vert g\left(t\right)\right\Vert_E\]
for any $\alpha>0$. Choosing $\alpha:=2\hat{K}$ we get
\[\left\Vert F_{t}(g)-F_{t}(\tilde{g})\right\Vert_{E}\leq\hat{K}\int_0^te^{2\hat{K}s}\left\Vert g-\tilde{g}\right\Vert_{*}ds\leq\frac{1}{2} e^{2\hat{K}t}\left\Vert g-\tilde{g}\right\Vert_{*}\]
and
\[\left\Vert F(g)-F(\tilde{g})\right\Vert_{*}\leq\frac{1}{2}\left\Vert g-\tilde{g}\right\Vert_{*}.\]
Therefore, by Banach's fixed point theorem there exists a unique fixed point $\hat{S}$ of $F$. As noted above $\hat{S}$ solves (\ref{system}). Moreover, the sequence of continuous time models defined via $\hat{S}^0\equiv s_0$ and $\hat{S}^{m+1}:=F\left(\hat{S}^m\right),\ m\in\N_0,$ converges to $\hat S$. We have the following result. 

\begin{thm}\label{T3} The fixed point $\hat{S}$ solves the ODE/PDE system (\ref{system}) and 
\[\lim_{m\ra\infty}\sup_{t\in[0,T]}\left\Vert\hat{S}^m(t)-\hat{S}(t)\right\Vert_E=0.\]
\end{thm}

The following lemma shows that $\hat{S}^m=\left(\hat{B}^m,\hat{v}^m\right)$ is Lipschitz continuous with respect to time.

\begin{lem}\label{Lip} There exists a constant $\hat{L}>0$ such that for all $m\in\N$ and all $s,t\in[0,T]$,
\[\left\Vert\hat{S}^m(t)-\hat{S}^m(s)\right\Vert_E\leq \hat{L}|t-s|. \]
\end{lem}

\begin{proof}
W.l.o.g.~$s<t$. By the mean value theorem there exists some $u\in(s,t)$ and for every $x\in\R$ a point $u_x\in(s,t)$ such that
\begin{eqnarray*}
\left|\hat{B}^m(t)-\hat{B}^m(s)\right|&=&\left|\left(\hat{B}^{m}\right)'(u)\right|(t-s)=\left|p^{B-A}\left[\hat{S}^{m-1}(u)\right]\right|(t-s)\leq t-s,\\
\left|\hat{v}^m(t,x)-\hat{v}^m(s,x)\right|&=&\left| \hat{v}_t^{m}(u_x,x)\right|(t-s)\leq\ol{L}(t-s),
\end{eqnarray*}
where the last inequality follows from Lemma \ref{dif}. 
Since the function $\hat{v}^m(t,\cdot)$ and its partial derivatives have compact support in $[-M-T,M+T]$,
\begin{eqnarray*}
\left\Vert\hat{v}^m(t)-\hat{v}^m(s)\right\Vert_{L^2}=\left\Vert\1_{[-M-T,M+T]}(\cdot)\left|\hat{v}_t^{m+1}(u_x,x)\right|(t-s)\right\Vert_{L^2}\leq \sqrt{2(M+T)}\cdot\ol{L}(t-s).
\end{eqnarray*}
Setting $\hat{L}:=1+\ol{L}\sqrt{2(M+T)}$ it follows that
\[\left\Vert\hat{S}^m(t)-\hat{S}^m(s)\right\Vert_E\leq \hat{L}|t-s| \quad\forall\ m\in\N.\]
\end{proof}

\subs{Uniqueness}

In order to show uniqueness of the solution to (\ref{system}) we assume to the contrary that there exists another solution $S$ and define the shifted volume density processes
\[\ol{v}(t,x):=\hat{v}\left(t,x-\int_0^tp^{B-A}\left[\hat{S}(s)\right]ds\right),\qquad \ul{v}(t,x):=v\left(t,x-\int_0^tp^{B-A}\left[S(s)\right]ds\right).\]
Then 
\begin{eqnarray*}
\ol{v}_t(t,x)&=&\hat{v}_t\left(t,x-\int_0^tp^{B-A}\left[\hat{S}(s)\right]ds\right)-p^{B-A}\left[\hat{S}(t)\right]\cdot\hat{v}_x\left(t,x-\int_0^tp^{B-A}\left[\hat{S}(s)\right]ds\right)\\
&=&f\left[\hat{S}(t)\right]\left(x-\int_0^tp^{B-A}\left[\hat{S}(s)\right]ds\right)
\end{eqnarray*}
and similarly for $\ul{v}$. Integrating with respect to $t$ and using that $\ol{v}(0)=\hat{v}(0)=v_0=v(0)=\ul{v}(0)$ we obtain from Assumption \ref{Lipschitz} and the mean value theorem,
\begin{flalign*}
\left\Vert\ol{v}(t)-\ul{v}(t)\right\Vert_{L^2}
& \leq\int_0^t\left\Vert f\left[\hat{S}(s)\right]\left(\cdot-\int_0^sp^{B-A}\left[\hat{S}(u)\right]du\right)-f\left[S(s)\right]\left(\cdot-\int_0^sp^{B-A}\left[{S}(u)\right]du\right)\right\Vert_{L^2} ds\\
& \leq\int_0^t\left(\left\Vert f\left[\hat{S}(s)\right]-f\left[S(s)\right]\right\Vert_{L^2} +C\int_0^s\left| p^{B-A}\left[\hat{S}(u)\right]-p^{B-A}\left[{S}(u)\right]\right|du\right)ds\\
&\leq  C\int_0^t \left\Vert\hat{S}(s)-S(s)\right\Vert_Eds.
\end{flalign*}

Moreover, as $|\hat{v}_x(t,\cdot)|$ is uniformly bounded by $\ol{J}$ with bounded support according to Lemma \ref{dif}, again by the mean value theorem
\begin{eqnarray*}
&&\left\Vert\hat{v}(t,\cdot)-{v}(t,\cdot)\right\Vert_{L^2}=
\left\Vert\hat{v}\left(t,\cdot-\int_0^tp^{B-A}\left[{S}(s)\right]ds\right)-{v}\left(t,\cdot-\int_0^tp^{B-A}\left[{S}(s)\right]ds\right)\right\Vert_{L^2}\\
&\leq&
\left\Vert\ol{v}(t,\cdot)-\ul{v}(t,\cdot)\right\Vert_{L^2}+
\left\Vert\hat{v}\left(t,\cdot-\int_0^tp^{B-A}\left[\hat{S}(s)\right]ds\right)-\hat{v}\left(t,\cdot-\int_0^tp^{B-A}\left[{S}(s)\right]ds\right)\right\Vert_{L^2}\\
&\leq& C\int_0^t \left\Vert \hat{S}\left(s\right)-S\left(s\right)\right\Vert_E ds+C
\left|\int_0^tp^{B-A}\left[\hat{S}(s)\right]ds-\int_0^tp^{B-A}\left[{S}(s)\right]ds\right| \\
& \leq & C\int_0^t \left\Vert \hat{S}\left(s\right)-S\left(s\right)\right\Vert_E ds.
\end{eqnarray*}
Furthermore,
\[\left|\hat{B}(t)-{B}(t)\right|\leq\int_0^t\left|p^{B-A}\left[\hat{S}(s)\right]-p^{B-A}\left[S(s)\right]\right|ds\leq 2L\int_0^t\left\Vert\hat{S}(s)-S(s)\right\Vert_Eds.\]
Therefore,
\[\left\Vert\hat{S}(t)-S(t)\right\Vert_E\leq C\int_0^t\left\Vert\hat{S}(s)-S(s)\right\Vert_Eds\]
and the continuous version of Gronwall's lemma, cf.~Lemma \ref{Gronwall2}, implies that $\hat{S}\equiv S$.

\s{Convergence of the discrete order book models}\label{proof}

The goal of this section is to prove Theorem \ref{limit}. Note that as opposed to \cite{HP} we cannot treat the price process independently of the volume densities because the conditional event probabilities and order placements resp.~cancelations do depend on both, prices and volumes. 

In the following we set for all $n\in\N$ and $k\in\R$,
\[\left(T_+^{(n)}\right)^k(f)(\cdot):=f\left(\cdot+k\Delta x^{(n)}\right).\]
Especially, this means that $\left(T_+^{(n)}\right)^{-1}\equiv T_-^{(n)}$. The following important Lemma  deals with multiple applications of the translation operator. It will be used repeatedly in what follows. 

\begin{lem}\label{counting}
If $f\in L^2$ is a step function in the $\Delta x^{(n)}$-grid, i.e.~$f$ satisfies $f(x)=f\left(l\Delta x^{(n)}\right)$ for all $x\in\left[l\Delta x^{(n)},(l+1)\Delta x^{(n)}\right)$ and all $l\in\N$, then for all $k,h\in\R$,
\[\left\Vert \left(\left(T_+^{(n)}\right)^{k}-\left(T_+^{(n)}\right)^{h}\right)(f)\right\Vert_{L^2}\leq  (|k-h|+1)\left\Vert \left( T_+^{(n)}-I\right)(f)\right\Vert_{L^2}.\]
\end{lem}

\begin{proof}
Because $\left(T^{(n)}_+\right)^{h}$ is an isometry, it is sufficient to prove the claim for $h=0$. 
To do this we first consider the case $k\in\N$ and claim that in this case even 
\[\left\Vert \left(\left(T_+^{(n)}\right)^k-I\right)(f)\right\Vert_{L^2}\leq  k\left\Vert \left( T_+^{(n)}-I\right)(f)\right\Vert_{L^2}.\]
Obviously, this is true for $k=1$. Assuming that it is true for $k-1$ we get
\begin{eqnarray*}
\left\Vert \left(\left(T^{(n)}_+\right)^k-I\right)(f)\right\Vert_{L^2}\leq\left\Vert T_+^{(n)}\circ\left(\left(T^{(n)}_+\right)^{k-1}-I\right)(f)\right\Vert_{L^2}+\left\Vert \left(T^{(n)}_+-I\right)(f)\right\Vert_{L^2}\\
=\left\Vert \left(\left(T^{(n)}_+\right)^{k-1}-I\right)(f)\right\Vert_{L^2}+\left\Vert \left(T^{(n)}_+-I\right)(f)\right\Vert_{L^2}\leq k\left\Vert \left(T^{(n)}_+-I\right)(f)\right\Vert_{L^2}
\end{eqnarray*}
and hence the inequality follows by induction for all $k\in\N$. Next for $k\in(0,1)$ either
\[\left(\left(T^{(n)}_+\right)^k-I\right)(f)(x)=\left(T_+^{(n)}-I\right)(f)(x)\qquad\text{or}\qquad\left(\left(T^{(n)}_+\right)^k-I\right)(f)(x)=0, \quad x\in\R_+.\]
Therefore, in this case
\[\left\Vert \left(\left(T^{(n)}_+\right)^k-I\right)(f)\right\Vert_{L^2}\leq\left\Vert \left(T^{(n)}_+-I\right)(f)\right\Vert_{L^2}.\]
Now take any $k\in\R_+$. Then,
\begin{eqnarray*}
\left\Vert \left(\left(T^{(n)}_+\right)^k-I\right)(f)\right\Vert_{L^2}&\leq&\left\Vert \left(\left(T^{(n)}_+\right)^{\lfloor k\rfloor}-I\right)(f)\right\Vert_{L^2}+\left\Vert \left(\left(T^{(n)}_+\right)^k-\left(T^{(n)}_+\right)^{\lfloor k\rfloor}\right)(f)\right\Vert_{L^2}\\
&\leq& \lfloor k\rfloor\left\Vert \left(T^{(n)}_+-I\right)(f)\right\Vert_{L^2}+\left\Vert \left(\left(T^{(n)}_+\right)^{k-\lfloor k\rfloor}-I\right)(f)\right\Vert_{L^2}\\
&\leq& (k+1)\left\Vert \left(T^{(n)}_+-I\right)(f)\right\Vert_{L^2}.
\end{eqnarray*}
Finally, the general case follows from the isometry property of the translation operator:
\[\left\Vert \left(\left(T^{(n)}_+\right)^k-I\right)(f)\right\Vert_{L^2}=\left\Vert \left(T^{(n)}_+\right)^{-k}\circ\left(\left(T^{(n)}_+\right)^k-I\right)(f)\right\Vert_{L^2}=\left\Vert \left(I-\left(T^{(n)}_+\right)^{-k}\right)(f)\right\Vert_{L^2}.\]
\end{proof}

\subs{A deterministic approximation of the discrete model for fixed $n\in\N$}

Recall that by definition
\begin{flalign*}
v^{(n)}_{k} =& \left(T_+^{(n)}\right)^{\sum_{j=0}^{k-1}\1_j^{(n),B-A}}\left(v_0^{(n)}\right)+\Delta v^{(n)}\sum_{j=0}^{k-1}\left(T_+^{(n)}\right)^{\sum_{i=j}^{k-1}\1_i^{(n),B-A}}\left(M_{j}^{(n)}\right).
\end{flalign*}

For each $n\in\N$ we are now going to define two approximations to the discrete model dynamics $S^{(n)}$, a deterministic non-linear approximation $\tilde{S}^{(n)}$ in which the event indicator variables are replaced by their averages conditioned on the state of the approximating sequence and a random approximation $\ol{S}^{(n)}$ in which the event indicator variables are replaced by their averages conditioned on the random state of the original state sequence. 

More precisely, we define for each $n\in\N$ the process $\tilde{S}^{(n)}$ through 
\[\tilde{B}^{(n)}_{k}:=B_0^{(n)}+\Delta x^{(n)}\sum_{j=0}^{k-1}p^{(n),B-A}\left[\tilde{S}^{(n)}_j\right]\]
and  
\begin{flalign*}
 \tilde{v}^{(n)}_{k}:=& \left(T_+^{(n)}\right)^{\sum_{j=0}^{k-1}p^{(n),B-A}\left[\tilde{S}^{(n)}_j\right]}\left(v_{0}^{(n)}\right)+\Delta v^{(n)} \sum_{j=0}^{k-1}\left(T_+^{(n)}\right)^{\sum_{i=j}^{k-1}p^{(n),B-A}\left[\tilde{S}^{(n)}_i\right]}\left(f^{(n)}\left[\tilde{S}^{(n)}_j\right]\right)
\end{flalign*}
for $k=0,1,\dots,\left\lfloor\frac{T}{\Delta t^{(n)}}\right\rfloor$.

Moreover, we define a second process $\ol{S}^{(n)}$ through
\[\ol{B}^{(n)}_{k}:=B^{(n)}_0+\Delta x^{(n)}\sum_{j=0}^{k-1}p^{(n),B-A}\left[{S}^{(n)}_j\right]\]
and
\begin{flalign*}
\ol{v}^{(n)}_{k}:=&\left(T_+^{(n)}\right)^{\sum_{j=0}^{k-1}p^{(n),B-A}\left[{S}^{(n)}_j\right]}\left(v_{0}^{(n)}\right)+
\Delta v^{(n)} \sum_{j=0}^{k-1}\left(T_+^{(n)}\right)^{\sum_{i=j}^{k-1}p^{(n),B-A}\left[{S}^{(n)}_i\right]}\left(f^{(n)}\left[{S}^{(n)}_j\right]\right)
\end{flalign*}
for $k=0,1,\dots,\left\lfloor\frac{T}{\Delta t^{(n)}}\right\rfloor$.

In a first step we are now going to show that the sequence $\ol{S}^{(n)}$ approximates $S^{(n)}$. The proof uses a weak law of large numbers for triangular martingale difference arrays, which can be found in the appendix. Subsequently we show the desired convergence of $\tilde{S}^{(n)}$ to the discrete model dynamics $S^{(n)}$.

In what follows $C>0$ denotes a generic constant that may vary from line to line and is independent of any index $h,i,j,k,l,n$.

\begin{thm}\label{T0}
Under the assumptions of Theorem \ref{limit} for all $\eps>0$,
\[\lim_{n\ra\infty}\p\left(\sup_{0\leq t\leq T}\left\Vert S^{(n)}(t)-\ol{S}^{(n)}(t)\right\Vert_E>\eps\right)=0.\]
\end{thm}

\begin{proof}
For the bid price component we have
\begin{eqnarray*}
\left|B^{(n)}_{k}-\ol{B}^{(n)}_{k}\right|=\Delta x^{(n)}\left|\sum_{j=0}^{k-1}\1^{(n),B}_j-\1^{(n),A}_j-p^{(n),B}\left[{S}^{(n)}_j\right]+p^{(n),A}\left[{S}^{(n)}_j\right]\right|.
\end{eqnarray*}
By definition the random variables
\[Y_j^{(n)}:=\Delta x^{(n)}\left(\1^{(n),B}_j-\1^{(n),A}_j-p^{(n),B}\left[{S}^{(n)}_j\right]+p^{(n),A}\left[{S}^{(n)}_j\right]\right),\quad j\leq\left\lfloor\frac{T}{\Delta t^{(n)}}\right\rfloor,\quad n\in\N,\]
form a triangular martingale difference array with respect to $\left(\F_j^{(n)}\right)$. If we can show that there exists $\alpha>\frac{1}{2}$ such that
\[\sup_{\substack{j\leq T/ \Delta t^{(n)}\\ n\in\N}}
\left(\frac{\E\left| Y^{(n)}_j\right|^2}{\left(\Delta t^{(n)}\right)^{2\alpha}}\right)<\infty,\]
then Theorem \ref{WLLN} will imply that 
\[\sup_{k\leq \frac{T}{\Delta t^{(n)}}}\left|\sum_{j=0}^{k-1} Y^{(n)}_j\right|=o(1)\quad\text{in probability.}\]
Indeed, this follows immediately from Assumption \ref{scaling} with $\alpha:=\frac{1+\beta}{2}$, because
\begin{eqnarray*}
\E\left|Y_{j}^{(n)}\right|^2&\leq&
4\left(\Delta x^{(n)}\right)^2\cdot\E\left(p^{(n),A}\left[{S}^{(n)}_{j}\right]+p^{(n),B}\left[{S}^{(n)}_{j}\right]\right)
\leq 4\left(\Delta x^{(n)}\right)^2\Delta p^{(n)}\leq4C\left(\Delta t^{(n)}\right)^{1+\beta}.
\end{eqnarray*}

Next we consider the volume component:
\begin{flalign*}
\left\Vert v^{(n)}_{k}-\ol{v}^{(n)}_{k}\right\Vert_{L^2}
\leq &\left\Vert\left(\left(T_+^{(n)}\right)^{\sum_{j=0}^{k-1}\1_j^{(n),B-A}}
-\left(T_+^{(n)}\right)^{\sum_{j=0}^{k-1}p^{(n),B-A}\left[{S}^{(n)}_j\right]}\right)\left(v_{0}^{(n)}\right)\right\Vert_{L^2}\\
&+\left\Vert\Delta v^{(n)} \sum_{j=0}^{k-1}\left(T_+^{(n)}\right)^{\sum_{i=j}^{k-1}\1_i^{(n),B-A}}\left(M_{j}^{(n)}-f^{(n)}\left[{S}^{(n)}_j\right]\right)\right\Vert_{L^2}\\
&+\Delta v^{(n)}\sum_{j=0}^{k-1} \left\Vert \left(\left(T_+^{(n)}\right)^{\sum_{i=j}^{k-1}\1_i^{(n),B-A}}- \left(T_+^{(n)}\right)^{\sum_{i=j}^{k-1}p^{(n),B-A}\left[{S}^{(n)}_i\right]}\right)\left(f^{(n)}\left[{S}^{(n)}_j\right]\right)\right\Vert_{L^2}.
\end{flalign*}

Let us first deal with the second term. Due to the norm invariance of the translation operator this term equals
\[\left\Vert\Delta v^{(n)} \sum_{j=0}^{k-1} \left(T_-^{(n)}\right)^{\sum_{i=0}^{j-1}\1_i^{(n),B-A}}\left(M_{j}^{(n)}-f^{(n)}\left[{S}^{(n)}_j\right]\right)\right\Vert_{L^2}.\]
The variables
\[
X_j^{(n)}:=\Delta v^{(n)}\left(T_-^{(n)}\right)^{\sum_{i=0}^{j-1}\1_i^{(n),B-A}}\left(M_{j}^{(n)}-f^{(n)}\left[{S}^{(n)}_j\right]\right)
\]
form a triangular martingale difference array with respect to $\left(\F_j^{(n)}\right)$. As $\left|M_j^{(n)}\right|$ is bounded by $M/\Delta x^{(n)}$ according to Assumption \ref{density},
\begin{eqnarray*}
\E\left\Vert X^{(n)}_j\right\Vert_{L^2}^2&\leq&4\left(\Delta v^{(n)}\right)^2\E\left\Vert\left(T_-^{(n)}\right)^{\sum_{i=0}^{j-1}\1_i^{(n),B-A}}\left(M_{j}^{(n)}\right)\right\Vert_{L^2}^2=4\left(\Delta v^{(n)}\right)^2\E\left\Vert \left(M_j^{(n)}\right)^2 \right\Vert_{L^1}\\
&\leq& 4M\frac{\left(\Delta v^{(n)}\right)^2}{\Delta x^{(n)}}\cdot\E\left\Vert M_j^{(n)}\right\Vert_{L^1}\leq 4M^2\frac{\left(\Delta v^{(n)}\right)^2}{\Delta x^{(n)}}.
\end{eqnarray*}
Therefore, for $\tilde{\alpha}:=\frac{2-\beta}{2}>\frac{1}{2}$ we have
\[\sup_{\substack{j\leq T/\Delta t^{(n)}\\ \ n\in\N}}\left(\frac{\E\left\Vert X_j^{(n)}\right\Vert_{L^2}^2}{\left(\Delta t^{(n)}\right)^{2\tilde{\alpha}}}\right)<\infty\]
and hence by Theorem \ref{WLLN},
\[\sup_{k\leq\left\lfloor\frac{T}{\Delta t^{(n)}}\right\rfloor}\left\Vert\sum_{j=0}^{k-1} X^{(n)}_j\right\Vert_{L^2}=o(1)\quad\text{in probability}.\]

Regarding the third term note that by Lemma \ref{counting}, Remark \ref{shiftf}, and Assumption \ref{scaling},
\begin{eqnarray*}
&&\Delta v^{(n)}\sum_{j=0}^{k-1} \left\Vert \left(\left(T_+^{(n)}\right)^{\sum_{i=j}^{k-1}\1_i^{(n),B-A}}- \left(T_+^{(n)}\right)^{\sum_{i=j}^{k-1}p^{(n),B-A}\left[{S}^{(n)}_i\right]}\right)\left(f^{(n)}\left[{S}^{(n)}_j\right]\right)\right\Vert_{L^2}\\
&=&\Delta v^{(n)}\sum_{j=0}^{k-1} \left\Vert \left(\left(T_+^{(n)}\right)^{\sum_{i=j}^{k-1}\1_i^{(n),B-A}-p^{(n),B-A}\left[{S}^{(n)}_i\right]}-I\right)\left(f^{(n)}\left[{S}^{(n)}_j\right]\right)\right\Vert_{L^2}\\
&\leq&\Delta v^{(n)}\sum_{j=0}^{k-1} \left(\left|\sum_{i=j}^{k-1}\1^{(n),B-A}_i-p^{(n),B-A}\left[S^{(n)}_i\right]\right|+1\right)\left\Vert\left(T_+^{(n)}-I\right)\left(f^{(n)}\left[{S}^{(n)}_j\right]\right)\right\Vert_{L^2}\\
&\leq& C\sup_{j\leq k-1} \left(\left|\sum_{i=j}^{k-1}\1^{(n),B-A}_i-p^{(n),B-A}\left[S^{(n)}_i\right]\right|+1\right)\Delta x^{(n)}.
\end{eqnarray*}
Thus, we may conclude as above for the price component that the term converges to zero in probability uniformly in $k\leq T/\Delta t^{(n)}$. The convergence of the first term in the above decomposition follows analogously.
\end{proof}

\begin{thm}\label{T1}
Under the assumptions of Theorem \ref{limit} for all $\eps>0$,
\[\lim_{n\ra\infty}\p\left(\sup_{0\leq t\leq T}\left\Vert S^{(n)}(t)-\tilde{S}^{(n)}(t)\right\Vert_E>\eps\right)=0.\]
\end{thm}

\begin{proof}
We have 
\[\left\Vert S^{(n)}\left(t^{(n)}_k\right)-\tilde{S}^{(n)}\left(t^{(n)}_k\right)\right\Vert_E\leq\left\Vert S^{(n)}\left(t^{(n)}_k\right)-\ol{S}^{(n)}\left(t^{(n)}_k\right)\right\Vert_E+\left\Vert\ol{S}^{(n)}\left(t^{(n)}_k\right)-\tilde{S}^{(n)}\left(t^{(n)}_k\right)\right\Vert_E.\]
According to Theorem \ref{T0} the first term converges uniformly to zero in probability. In the following we use the Lipschitz continuity of the $p^I$s and $f^{(n)}$s formulated in Assumption \ref{Lipschitz} to derive an appropriate upper bound for the second term, which allows us to apply the discrete Gronwall lemma in order to prove the assertion. We start again with the bid price component:
\[\left|\ol{B}^{(n)}_{k}-\tilde{B}^{(n)}_{k}\right|=\Delta x^{(n)}\left|\sum_{j=0}^{k-1} p^{(n),B-A}\left[S_j^{(n)}\right]-p^{(n),B-A}\left[\tilde{S}_j^{(n)}\right]\right|
\leq\Delta x^{(n)}\Delta p^{(n)} 2L \sum_{j=0}^{k-1}\left\Vert S^{(n)}_j-\tilde{S}^{(n)}_j\right\Vert_E.\]

Moreover, using Lemma \ref{counting} and Assumption \ref{Lipschitz},
\begin{eqnarray*}
\Delta v^{(n)} \sum_{j=0}^{k-1}\left\Vert \left(T_+^{(n)}\right)^{\sum_{i=j}^{k-1}p^{(n),B-A}\left[\tilde{S}^{(n)}_j\right]}\left(f^{(n)}\left[\tilde{S}^{(n)}_j\right]-f^{(n)}\left[{S}^{(n)}_j\right]\right)\right\Vert_{L^2}\leq\Delta v^{(n)} L\sum_{j=0}^{k-1}\left\Vert \tilde{S}^{(n)}_j-S^{(n)}_j\right\Vert_E.
\end{eqnarray*}

Similarly to the proof of Theorem \ref{T0}, using Lemma \ref{counting} and Remark \ref{shiftf} we can derive the inequality
\begin{eqnarray*}
&&\left\Vert \left(\left(T_+^{(n)}\right)^{\sum_{i=j}^{k-1}p^{(n),B-A}\left[\tilde{S}^{(n)}_i\right]}-\left(T_+^{(n)}\right)^{\sum_{i=j}^{k-1}p^{(n),B-A}\left[{S}^{(n)}_i\right]}\right)\left(f^{(n)}\left[{S}^{(n)}_j\right]\right)\right\Vert_{L^2}\\
&\leq&C\left(\left|\sum_{i=j}^{k-1}p^{(n),B-A}\left[\tilde{S}^{(n)}_i\right]-p^{(n),B-A}\left[S^{(n)}_i\right]\right|+1\right)\Delta x^{(n)}\\
&\leq&C\Delta x^{(n)}\Delta p^{(n)} 2L \sum_{i=0}^{k-1}\left\Vert S^{(n)}_i-\tilde{S}^{(n)}_i\right\Vert_E+C\Delta x^{(n)}.
\end{eqnarray*}
Relying on Assumption \ref{initial} instead of Assumption \ref{Lipschitz} we may replace $f^{(n)}\left[{S}^{(n)}_j\right]$ by $v^{(n)}_{0}$ in the above computations to get a similar estimate for the initial volume term.

Finally, putting everything together we have by Assumption \ref{scaling},
\begin{align*}
&\left\Vert\tilde{v}^{(n)}_{k}-\ol{v}^{(n)}_{k}\right\Vert_{L^2}\leq\left\Vert\left(\left(T_+^{(n)}\right)^{\sum_{j=0}^{k-1}p^{(n),B-A}\left[\tilde{S}^{(n)}_j\right]}-\left(T_+^{(n)}\right)^{\sum_{j=0}^{k-1}p^{(n),B-A}\left[{S}^{(n)}_j\right]}\right)\left(v_{0}^{(n)}\right)\right\Vert_{L^2}\\
&+ \Delta v^{(n)} \sum_{j=0}^{k-1}\left\Vert\left(T_+^{(n)}\right)^{\sum_{i=j}^{k-1}p^{(n),B-A}\left[\tilde{S}^{(n)}_i\right]}\left(f^{(n)}\left[\tilde{S}^{(n)}_j\right]-f^{(n)}\left[{S}^{(n)}_j\right]\right)\right\Vert_{L^2}\\
&+\Delta v^{(n)}\sum_{j=0}^{k-1} \left\Vert \left(\left(T_+^{(n)}\right)^{\sum_{i=j}^{k-1}p^{(n),B-A}\left[\tilde{S}^{(n)}_i\right]}- \left(T_+^{(n)}\right)^{\sum_{i=j}^{k-1}p^{(n),B-A}\left[{S}^{(n)}_i\right]}\right)\left(f^{(n)}\left[{S}^{(n)}_j\right]\right)\right\Vert_{L^2}\\
&\leq C\left(\Delta t^{(n)}\sum_{l=0}^{k-1}\left\Vert \tilde{S}^{(n)}_l-S^{(n)}_l\right\Vert_E+\Delta x^{(n)}+\Delta t^{(n)}\sum_{j=0}^{k-1}\left[\Delta t^{(n)}\sum_{l=0}^{k-1}\left\Vert \tilde{S}^{(n)}_l-S^{(n)}_l\right\Vert_E+\Delta x^{(n)}\right]\right)\\
&\leq C\Delta t^{(n)}\sum_{l=0}^{k-1}\left\Vert \tilde{S}^{(n)}_l-S^{(n)}_l\right\Vert_E+C\Delta x^{(n)}.
\end{align*}
Therefore, for some sequence $(a_n)$ converging to zero in probability and for all $k\leq \frac{T}{\Delta t^{(n)}}$ we get the following uniform estimate by means of the discrete Gronwall Lemma \ref{Gronwall}:
\begin{eqnarray*}
\left\Vert S^{(n)}_{k}-\tilde{S}_{k}^{(n)}\right\Vert_E&\leq&\left\Vert S^{(n)}_{k}-\ol{S}_{k}^{(n)}\right\Vert_E+\left\Vert\ol{S}^{(n)}_{k}-\tilde{S}_{k}^{(n)}\right\Vert_E\leq
 a_n+C\Delta t^{(n)}\sum_{l=0}^{k-1}\left\Vert \tilde{S}^{(n)}_l-S^{(n)}_l\right\Vert_E\\
&\leq& a_n+a_nC\Delta t^{(n)}\sum_{l=0}^{k-1}e^{\sum_{m=l+1}^{k-1}C\Delta t^{(n)}}\leq C a_n.
\end{eqnarray*}
\end{proof}

\subs{Almost uniform iteration to the discrete approximation}\label{fixd}

In this section we approximate for each $n\in\N$ the model $\tilde{S}^{(n)}$ iteratively by an $\tilde{E}$-valued sequence $\left(\tilde{S}^{(n),m}\right)_m$ of limit order book models. To this end, we define for each $n\in\N$ a function $F^{(n)}:\ \tilde{E}\ra\tilde{E}$ via 
\[F^{(n)}:\ g\mapsto G^{(n)}=\left(G^{(n)}_B,G^{(n)}_v\right)\quad \text{with}\quad G^{(n)}(t):=G^{(n)}_j,\ \text{if }t\in\left[t^{(n)}_j,t^{(n)}_{j+1}\right),\]
where $G^{(n)}_0:=s_0^{(n)}$ and for $k\in\N$,
\[G^{(n)}_{B,k}:=G^{(n)}_{B,k-1}+\Delta x^{(n)}p^{(n),B-A}\left[g^{(n)}_{k-1}\right]\]
as well as
\begin{align*}
G^{(n)}_{v,k}:=\left(T_+^{(n)}\right)^{\sum_{j=0}^{k-1}p^{(n),B-A}\left[g^{(n)}_j\right]}\left(v_{0}^{(n)}\right)+\Delta v^{(n)} \sum_{j=0}^{k-1}\left(T_+^{(n)}\right)^{\sum_{i=j}^{k-1}p^{(n),B-A}\left[g^{(n)}_i\right]}\left(f^{(n)}\left[g^{(n)}_j\right]\right)
\end{align*}
with
\[g_j^{(n)}:=g\left(t^{(n)}_j\right).\]
We will write $F^{(n)}_{k}(g):=G^{(n)}_k$ for $G^{(n)}_k$ defined as above in what follows.\\ 

Using analogous arguments as in the proof of Theorem \ref{T1} one can find a constant $K>0$ such that
\[\left\Vert F_{k+1}^{(n)}(g)-F_{k+1}^{(n)}(\tilde{g})\right\Vert_E\leq K\Delta v^{(n)}\sum_{j=0}^k\left\Vert g_j-\tilde{g}_j\right\Vert_E+K\Delta x^{(n)}\quad\text{for all }k\leq \frac{T}{\Delta t^{(n)}},\ n\in\N.\]
Now as in Section \ref{fixc} we define a weighted norm on $\tilde{E}$ via
\[\left\Vert g\right\Vert_{**}:=\sup_{0\leq t\leq T}e^{-3Kt}\left\Vert g(t)\right\Vert_E,\]
which allows us to get an estimate for the weighted norm with Lipschitz constant less than $1$ up to an error of order $\Delta x^{(n)}$. As for all $n\in\N$,
\begin{eqnarray*}
\left\Vert F_{k+1}^{(n)}(g)-F_{k+1}^{(n)}(\tilde{g})\right\Vert_E&\leq& K\Delta v^{(n)}\sum_{j=0}^k e^{3K t^{(n)}_j}\left\Vert g-\tilde{g}\right\Vert_{**}+K\Delta x^{(n)}\\
&=&K\Delta v^{(n)}\cdot\frac{e^{3K(k+1)\Delta t^{(n)}}-1}{e^{3K\Delta t^{(n)}}-1}\cdot\left\Vert g-\tilde{g}\right\Vert_{**}+K\Delta x^{(n)},
\end{eqnarray*}
there exists by Assumption \ref{scaling} an $N_0\in\N$ such that for all $n\geq N_0$,
\[e^{-3K t^{(n)}_{k+1}}\left\Vert F_{k+1}^{(n)}(g)-F_{k+1}^{(n)}(\tilde{g})\right\Vert_E\leq\frac{K\Delta v^{(n)}}{e^{3K\Delta t^{(n)}}-1}\left\Vert g-\tilde{g}\right\Vert_{**}+K\Delta x^{(n)}\leq\frac{1}{2}\left\Vert g-\tilde{g}\right\Vert_{**}+K\Delta x^{(n)}\]
and therefore indeed
\[\left\Vert F^{(n)}(g)-F^{(n)}(\tilde{g})\right\Vert_{**}\leq\frac{1}{2}\left\Vert g-\tilde{g}\right\Vert_{**}+K\Delta x^{(n)}\qquad\forall\ n\geq N_0.\] 
W.l.o.g.~we take $N_0=1$ in the following. For each $n,m\in\N_0$ we define a new discrete time model $\tilde{S}^{(n),m}$ via $\tilde{S}^{(n),0}(t)\equiv s_0^{(n)}$ as well as 
\[\tilde{S}^{(n),m+1}:=F^{(n)}\left(\tilde{S}^{(n),m}\right).\]

\begin{thm}\label{T2}
\[\lim_{\substack{m\ra\infty\\n\ra\infty}}\sup_{0\leq t\leq T}\left\Vert\tilde{S}^{(n),m}(t)-\tilde{S}^{(n)}(t)\right\Vert_E=0.\]
\end{thm}

\begin{proof}
First, note that by definition $\tilde{S}^{(n)}$ is a fixed point of $F^{(n)}$, i.e.
\[F^{(n),m}\left(\tilde{S}^{(n)}\right):=\underbrace{F^{(n)}\circ\dots\circ F^{(n)}}_{m\text{ times}}\left(\tilde{S}^{(n)}\right)=\tilde{S}^{(n)}\quad\forall\ m\in\N. \]
Hence, making use of the above computations we deduce that for every $m,n\in\N$,
\begin{eqnarray*}
\left\Vert \tilde{S}^{(n),m}-\tilde{S}^{(n)}\right\Vert_{**}&=&\left\Vert F^{(n),m}\left(\tilde{S}^{(n),0}\right)-F^{(n),m}\left(\tilde{S}^{(n)}\right)\right\Vert_{**}\\
&\leq&\left(\frac{1}{2}\right)^{m}\left\Vert\tilde{S}^{(n),0}-\tilde{S}^{(n)}\right\Vert_{**}+\sum_{k=0}^{m-1}\left(\frac{1}{2}\right)^kK\Delta x^{(n)}\\
&\leq&\left(\frac{1}{2}\right)^{m}\sup_{0\leq t\leq T}\left\Vert s_0^{(n)}-\tilde{S}^{(n)}(t)\right\Vert_E+2K\Delta x^{(n)}\\
&\leq&\left(\frac{1}{2}\right)^{m-1}\ol{K}+2K\Delta x^{(n)}.
\end{eqnarray*}
Now the result follows from the equivalence of the weighted norm and the norm $\sup_{0\leq t\leq T}\left\Vert\cdot\right\Vert_E$.
\end{proof}

Especially, Theorem \ref{T2} implies that
\[\lim_{m\ra\infty}\lim_{n\ra\infty}\sup_{0\leq t\leq T}\left\Vert\tilde{S}^{(n),m}(t)-\tilde{S}^{(n)}(t)\right\Vert_E=\lim_{n\ra\infty}\lim_{m\ra\infty}\sup_{0\leq t\leq T}\left\Vert\tilde{S}^{(n),m}(t)-\tilde{S}^{(n)}(t)\right\Vert_E=0.\]

\subs{Convergence of the discrete iteration to the continuous iteration}

The goal of this section is to prove the following result.

\begin{thm}\label{T4} For all $m\in\N$,
\[\lim_{n\ra\infty}\sup_{0\leq t\leq T}\left\Vert \tilde{S}^{(n),m}(t)-\hat{S}^m(t)\right\Vert_E=0.\]
\end{thm}

\begin{proof}
To prove Theorem \ref{T4} we proceed by induction. Obviously, for $m=0$ the claim holds by Assumption \ref{initial}. Now assume that the claim holds for $m$ and consider the $(m+1)$-th iteration: 
First, we show the convergence of the bid price process. Writing the integral as a limit of Riemann sums we have by Assumptions \ref{initial} and \ref{scaling},
\begin{eqnarray*}
\hat{B}^{m+1}(t)=\lim_{n\ra\infty}B_0^{(n)}+\lim_{n\ra\infty}\sum_{j=0}^{\left\lfloor t/\Delta t^{(n)}\right\rfloor}\Delta x^{(n)} \Delta p^{(n)}p^{B-A}\left[\hat{S}^{m}\left(t^{(n)}_j\right)\right].
\end{eqnarray*}
As $p^{B-A}$ and $\hat{S}^m$ are both Lipschitz continuous by Assumption \ref{Lipschitz} and Lemma \ref{Lip}, this convergence is uniform in $t\in[0,T]$. Moreover by Assumption \ref{Lipschitz},
\begin{eqnarray*}
&&\sup_{0\leq t\leq T}\left|\sum_{j=0}^{\left\lfloor t/\Delta t^{(n)}\right\rfloor}\Delta x^{(n)} \Delta p^{(n)}\left(p^{B-A}\left[\hat{S}^{m}\left(t^{(n)}_j\right)\right]-p^{B-A}\left[\tilde{S}^{(n),m}_j\right]\right)\right|\\
&\leq& T\cdot\frac{\Delta x^{(n)} \Delta p^{(n)}}{\Delta t^{(n)}}\cdot\sup_{j\leq\frac{T}{\Delta t^{(n)}}}\left|p^{B-A}\left[\hat{S}^{m}\left(t^{(n)}_j\right)\right]-p^{B-A}\left[\tilde{S}^{(n),m}_j\right]\right|\\
&\leq& 2LT\cdot\frac{\Delta x^{(n)} \Delta p^{(n)}}{\Delta t^{(n)}}\cdot\sup_{0\leq t\leq T}\left\Vert\hat{S}^{m}(t)-\tilde{S}^{(n),m}(t)\right\Vert_E,
\end{eqnarray*}
which converges towards zero as $n\ra\infty$ by Assumption \ref{scaling} and the induction hypothesis.

We now show the convergence of the buy side volume density function step by step. W.l.o.g.~we only prove the convergence of the order placement / cancelation term. The convergence of the term involving the initial volume density function  follows by analogous arguments. First note that a pointwise Riemann sum approximation gives
\begin{align*}
&\int_0^t f\left[\hat{S}^{m}(s)\right]\left(x+\int_s^tp^{B-A}\left[\hat{S}^{m}(u)\right]du\right)ds\\
&\qquad=\lim_{n\ra\infty}\Delta t^{(n)}\sum_{j=0}^{\lfloor t/\Delta t^{(n)}\rfloor}f\left[\hat{S}^{m}\left(t^{(n)}_j\right)\right]\left(x+\int_{t^{(n)}_j}^tp^{B-A}\left[\hat{S}^{m}(u)\right]du\right).
\end{align*}
To show that the convergence also holds in $L^2$ observe that for all $\ul{t}\leq\ol{t}\in[0,T]$ with $|\ol{t}-\ul{t}|\leq\Delta t^{(n)}$, making use of Assumption \ref{Lipschitz}, Lemma \ref{Lip}, and the mean value theorem,
\begin{flalign*}
&\left\Vert f\left[\hat{S}^{m}\left(\ol{t}\right)\right]\left(\cdot+\int_{\ol{t}}^tp^{B-A}\left[\hat{S}^{m}(u)\right]du\right)-f\left[\hat{S}^{m}\left(\ul{t}\right)\right]\left(\cdot+\int_{\ul{t}}^tp^{B-A}\left[\hat{S}^{m}(u)\right]du\right)\right\Vert_{L^2}\\
&\leq\left\Vert f\left[\hat{S}^{m}\left(\ol{t}\right)\right]-f\left[\hat{S}^{m}\left(\ul{t}\right)\right]\right\Vert_{L^2}+C\int_{\ul{t}}^{\ol{t}}\left|p^{B-A}\left[\hat{S}^{m}(u)\right]\right|du\leq L\hat{L}|\ol{t}-\ul{t}|+C|\ol{t}-\ul{t}|\leq C\Delta t^{(n)}.
\end{flalign*}
Therefore, the above convergence does indeed hold in $L^2$, uniformly in $t\in[0,T]$ .

Second, by similar arguments\small
\begin{flalign*}
&\Delta t^{(n)}\sum_{j=0}^{\lfloor t/\Delta t^{(n)}\rfloor}\left\Vert \left(T_+^{(n)}\right)^{\sum_{i=j}^{\left\lfloor t/\Delta t^{(n)}\right\rfloor}p^{(n),B-A}\left[\tilde{S}^{(n),m}_i\right]}\left(f\left[\hat{S}^{m}\left(t^{(n)}_j\right)\right]\right)-f\left[\hat{S}^{m}\left(t^{(n)}_j\right)\right]\left(\cdot+\int_{t^{(n)}_j}^tp^{B-A}\left[\hat{S}^{m}(u)\right]du\right)\right\Vert_{L^2}\\
&\leq\sup_{j\leq t/\Delta t^{(n)}}\left\Vert \sup_{z\in\R}\left|\left(f\left[\hat{S}^m\left(t_j^{(n)}\right)\right]\right)'(z)\right|\1_{[-M,M]}(\cdot)\left(\Delta x^{(n)}\sum_{i=j}^{\left\lfloor t/\Delta t^{(n)}\right\rfloor}p^{(n),B-A}\left[\tilde{S}^{(n),m}_i\right]-\int_{t_j^{(n)}}^tp^{B-A}\left[\hat{S}^{m}(u)\right]du\right)\right\Vert_{L^2}\\
&\leq C\sup_{j\leq t/\Delta t^{(n)}}\left|\Delta x^{(n)}\Delta p^{(n)}\sum_{i=j}^{\left\lfloor t/\Delta t^{(n)}\right\rfloor}p^{B-A}\left[\tilde{S}^{(n),m}_i\right]-\int_{t_j^{(n)}}^tp^{B-A}\left[\hat{S}^{m}(u)\right]du\right|
\end{flalign*}\normalsize
and as for the price component this term converges to zero uniformly in $t\in[0,T]$. 

Third, note that by Lemma \ref{counting} and Assumption \ref{Lipschitz} we have
\begin{flalign*}
&\Delta t^{(n)}\sum_{j=0}^{\lfloor t/\Delta t^{(n)}\rfloor}\left\Vert\left(T_+^{(n)}\right)^{\sum_{i=j}^{\left\lfloor t/\Delta t^{(n)}\right\rfloor}p^{(n),B-A}\left[\tilde{S}^{(n),m}_i\right]}\left(f\left[\hat{S}^{m}\left(t^{(n)}_j\right)\right]-f^{(n)}\left[\tilde{S}^{(n),m}_j\right]\right)\right\Vert_{L^2}\\
&\leq T \sup_{j\leq t/\Delta t^{(n)}}\left\{\left\Vert f\left[\hat{S}^{m}\left(t^{(n)}_j\right)\right]-f^{(n)}\left[\hat{S}^{m}\left(t^{(n)}_j\right)\right]\right\Vert_{L^2}+\left\Vert f^{(n)}\left[\hat{S}^{m}\left(t^{(n)}_j\right)\right]-f^{(n)}\left[\tilde{S}^{(n),m}_j\right]\right\Vert_{L^2}\right\}\\
&\leq o(1)+L \sup_{j\leq T/\Delta t^{(n)}}\left\Vert\hat{S}^{m}\left(t^{(n)}_j\right)-\tilde{S}^{(n),m}\left(t^{(n)}_j\right)\right\Vert_E,
\end{flalign*}
which converges to zero by the induction hypothesis.

Therefore, we proved that uniformly in $t\in[0,T]$ the term
\[\Delta t^{(n)}\sum_{j=0}^{\lfloor t/\Delta t^{(n)}\rfloor}\left(T_+^{(n)}\right)^{\sum_{i=j}^{\lfloor t/\Delta t^{(n)}\rfloor}p^{(n),B-A}\left[\tilde{S}^{(n),m}_i\right]}\left(f^{(n)}\left[\tilde{S}^{(n),m}_j\right]\right)\]
converges in $L^2$ towards
\[\int_0^t f\left[\hat{S}^{m}(s)\right]\left(\cdot+\int_s^tp^{B-A}\left[\hat{S}^{m}(u)\right]du\right)ds.\]

\end{proof}

\subs{Proof of Theorem \ref{limit}}

Finally, let us put the partial convergence results proven in the previous subsections together to prove the convergence of the discrete limit order book models $S^{(n)}$ to $\hat{S}$. For this fix $\eps>0$. Then there exists by Theorem \ref{T3} an $M_1=M_1(\eps)$ such that for all $m\geq M_1$,
\[\left\Vert\hat{S}^m-\hat{S}\right\Vert_E<\frac{\eps}{6}.\]
Also, by Theorem \ref{T2} there exist $M_2=M_2(\eps)$ and $N_1=N_1(\eps)$ such that for all $m\geq M_2$ and $n\geq N_1$,
\[\left\Vert\tilde{S}^{(n)}-\tilde{S}^{(n),m}\right\Vert_E<\frac{\eps}{6}.\]
We set $M_0=M_0(\eps):=M_1(\eps)\vee M_2(\eps)$. Then for all $n\geq N_1$,
\[\left\Vert\hat{S}^{M_0}-\hat{S}\right\Vert_E+\left\Vert\hat{S}^{(n)}-\hat{S}^{(n),M_0}\right\Vert_E<\frac{\eps}{3}.\]
Furthermore, Theorem \ref{T4} yields the existence of an $N_2=N_2(M_0,\eps)=N_2(M_0(\eps),\eps)=N_2(\eps)$ such that for all $n\geq N_2$,
\[\left\Vert\tilde{S}^{(n),M_0}-\hat{S}^{M_0}\right\Vert_E<\frac{\eps}{6}.\]
Hence, for all $n\geq N_0=N_0(\eps):=N_1(\eps)\vee N_2(\eps)$,
\[\left\Vert\tilde{S}^{(n)}-\hat{S}\right\Vert_E\leq\left\Vert\tilde{S}^{(n)}-\tilde{S}^{(n),M_0}\right\Vert_E+\left\Vert\tilde{S}^{(n),M_0}-\hat{S}^{M_0}\right\Vert_E+\left\Vert\hat{S}^{M_0}-\hat{S}\right\Vert_E<\frac{\eps}{2}.\]
Finally, by Theorem \ref{T1} there exists for every $\delta \in(0,1)$ an $N_3=N_3(\eps,\delta)$ such that for all $n\geq N_3$,
\[\p\left(\left\Vert S^{(n)}-\tilde{S}^{(n)}\right\Vert_E\leq \frac{\eps}{2}\right)>1-\delta.\]
Setting $N=N(\eps,\delta):=N_0(\eps)\vee N_3(\eps,\delta)$ we conclude that for all $n\geq N$,
\[\p\left(\left\Vert S^{(n)}-\hat{S}\right\Vert_E\leq \eps\right)>1-\delta\quad\LRA\quad\p\left(\left\Vert S^{(n)}-\hat{S}\right\Vert_E> \eps\right)<\delta.\]
As $\eps,\delta>0$ were arbitrary, this completes the convergence proof.

 \hfill $\Box$

\s{The two-sided limit order book model}\label{2}

For the ease of notation we have concentrated on the one-sided LOB model in the previous sections. It is however straight forward to generalize the proof of Theorem \ref{limit} from a one-sided to a two-sided LOB model, in which both sides are modeled in a similar manner. In what follows, we briefly formulate the assumptions and dynamics underlying the two-sided model and state the corresponding convergence result. The proof will be omitted as it is works in exactly the same way as the proof of the one-sided model discussed in detail in the previous sections. To illustrate the usefulness of the general state dependency, Subsection \ref{sim} contains some simulation results of the two-sided LOB model.

\subs{Setup and convergence result}

In this section the state of the order book in the $n$-th model is represented by 
\[S^{(n)}(t):=\left(A^{(n)}(t),v^{(n)}_a(t),B^{(n)}(t),v^{(n)}_b(t)\right),\quad t\in[0,T],\]
where $A^{(n)}(t)$ describes the best ask price, $B^{(n)}(t)$ describes the best bid price, $v^{(n)}_a(t,\cdot)$ represents the sell side volume density function and $v^{(n)}_b(t,\cdot)$ the buy side volume density function (both in relative coordinates) at time $t$. Here, the volumes of the real order book correspond to $v_b^{(n)}(t,x),\ x\leq 0,$ and $v_a^{(n)}(t,x),\ x\geq0$, while the extension of $v_b^{(n)}(t)$ to the positive halfline and the extension of $v_a^{(n)}(t)$ to the negative halfline will be understood as the respective shadow book of each side and are used to model the distribution of order placements inside the spread. 

The event variables $\phi_k^{(n)}$ will now take their values in the set $\{A,B,C,D,E,F\}$ with
\bi
\item A = market sell order
\item B = buy limit order inside the spread
\item C = placement or cancelation of a buy limit order
\item D = market buy order
\item E = sell limit order inside the spread
\item F = placement or cancelation of a sell limit order.
\ei

Here the buy side events $A,B,C$ are defined in the same way as before, while the sell side events $D,E,F$ are symmetric copies of the respective buy side events. 
Hence, the dynamics of the two sided order book can formally be written down as
\[S^{(n)}(t)=S^{(n)}_k:=\left(A^{(n)}_k,v^{(n)}_{a,k},B^{(n)}_k,v^{(n)}_{b,k}\right),\quad t\in\left[t_k^{(n)},t^{(n)}_{k+1}\right),\]
where for $k=1,\dots,\lfloor T/\Delta t^{(n)}\rfloor$,
\begin{eqnarray*}
A_k^{(n)}&=&A_{k-1}^{(n)}+\1_{k-1}^{(n),D-E}\Delta x^{(n)},\\
B_k^{(n)}&=&B_{k-1}^{(n)}+\1_{k-1}^{(n),B-A}\Delta x^{(n)},\\
v_{a,k}^{(n)}&=&v_{a,k-1}^{(n)}+\1^{(n),D}_{k-1}\left(T^{(n)}_+-I\right)\left(v^{(n)}_{a,k-1}\right)+\1^{(n),E}_{k-1}\left(T^{(n)}_--I\right)\left(v^{(n)}_{a,k-1}\right)\\
&&\qquad\qquad+\1^{(n),F}_{k-1}\frac{\Delta v^{(n)}}{\Delta x^{(n)}}\omega_{k-1}^{(n)}\sum_j\1_{\left\{\pi_{k-1}^{(n)}\in\left[x_j^{(n)},x^{(n)}_{j+1}\right)\right\}}(\cdot),\\
v_{b,k}^{(n)}&=&v_{b,k-1}^{(n)}+\1^{(n),B}_{k-1}\left(T^{(n)}_+-I\right)\left(v^{(n)}_{b,k-1}\right)+\1^{(n),A}_{k-1}\left(T^{(n)}_--I\right)\left(v^{(n)}_{b,k-1}\right)\\
&&\qquad\qquad+\1^{(n),C}_{k-1}\frac{\Delta v^{(n)}}{\Delta x^{(n)}}\omega_{k-1}^{(n)}\sum_j\1_{\left\{\pi_{k-1}^{(n)}\in\left[x_j^{(n)},x^{(n)}_{j+1}\right)\right\}}(\cdot).
\end{eqnarray*}

Besides the scaling assumption \ref{scaling} we make the following assumption on the two-sided limit order book model in analogy to Assumptions \ref{initial}, \ref{density}, and \ref{Lipschitz} before.

\begin{ass}\label{ass2}\mbox{}
\begin{enumerate}
\item 
There exist $A_0,B_0\in\R_+$ such that $A^{(n)}_0\ra A_0$ and $B^{(n)}_0\ra B_0$. Moreover, 
the initial volume density functions $v^{(n)}_{a,0}$ and $v^{(n)}_{b,0}$ are non-negative step-functions on the grid \mbox{$\{x_j^{(n)},\ j\in\Z\}$,} which are uniformly bounded by $M$ and have compact support in $[-M,M]$ for all $n\in\N$. There exist non-negative continuously differentiable functions $v_{a,0},v_{b,0}\in L^2$ such that
\[\left\Vert v_{i,0}^{(n)}-v_{i,0}\right\Vert_{L^2}=\mathcal{O}\left(\Delta x^{(n)}\right),\quad i=a,b.\]
\item The random variables $\left(\omega_k^{(n)}\right)_{k,n\in\N_0}$ and $\left(\pi_k^{(n)}\right)_{k,n\in\N_0}$ take their values in the compact interval $[-M,M]$ almost surely.
\item There are Lipschitz continuous functions $p^A,p^B,p^D,p^E:E\times E\ra[0,1]$ with Lipschitz constant $L$ and a scaling parameter $\Delta p^{(n)}$ such that for all $n\in\N_0$ and $k\leq \lfloor T/\Delta t^{(n)}\rfloor$,
\begin{eqnarray*}
\p\left(\left.\phi^{(n)}_k=I \ \right|S_j^{(n)},\ j\leq k\right)=\Delta p^{(n)} p^I\left[S_k^{(n)}\right] \quad a.s.\quad\text{for}\ I=A,B,D,E.
\end{eqnarray*}
\item There are Lipschitz continuous functions $f_C^{(n)},f^{(n)}_F:E\times E\ra L^2,\ n\in\N_0,$ with common Lipschitz constant $L>0$ such that for all $k\leq \lfloor T/\Delta t^{(n)}\rfloor$ and $I=C,F$,
\[
	f^{(n)}_I\left[S_k^{(n)}\right](\cdot)=\frac{1}{\Delta x^{(n)}}\E\left(\left. \omega_k^{(n)}\sum_{j\in\Z} \1_{\left\{\pi_k^{(n)}\in\left[x_j^{(n)},x^{(n)}_{j+1}\right)\right\}}(\cdot)\1_I\left(\phi_k^{(n)}\right)\right|S_j^{(n)},\ j\leq k\right)\quad a.s.
	\]
and
\[\sup_{s\in E}\left\Vert f_I^{(n)}[s](\cdot)\right\Vert_\infty\leq M.\]
Moreover, there exist two functions $f_C,f_F:E\times E\ra L^2$ such that for $I=C,F$,
\[\sup_{s\in E}\left\Vert f^{(n)}_I[s]-f_I[s]\right\Vert_{L^2}=\mathcal{O}\left(\Delta x^{(n)}\right),\]
where each $f_I[s](\cdot):\R\ra[-M,M]$ is continuously differentiable in $x$ for all $s\in E'\times E'$ with derivate being uniformly bounded in absolute value by $M$.
\end{enumerate}
\end{ass}

In the following the norm on $E\times E$ is defined as $\left\Vert (s_1,s_2)\right\Vert_{E\times E}:=\left\Vert s_1\right\Vert_E+\left\Vert s_2\right\Vert_E$ for all $s_1,s_2\in E$. Moreover, we set $s_0:=\left(A_0,v_{a,0},B_0,v_{b,0}\right)\in E\times E$.

\begin{thm}\label{limit2}
Under Assumptions \ref{ass2} and \ref{scaling} there exists a deterministic process $S:[0,T]\ra E\times E$ such that for all $\eps>0$,
\[\lim_{n\ra\infty}\p\left(\sup_{0\leq t\leq T}\left\Vert S^{(n)}(t)-S(t)\right\Vert_{E\times E}>\eps\right)=0. \]
Moreover, $S=(A,v_a,B,v_b)$ is the unique classical solution to the following coupled ODE/PDE initial boundary value problem:
\begin{equation*}
\begin{aligned}
S(0)&=s_0,\\
dA(t)&=p^{D-E}[S(t)]dt,\quad t\in[0,T],\\
dB(t)&=p^{B-A}[S(t)]dt,\quad t\in[0,T],\\
\partial_tv_a(t,x)&=p^{D-E}[S(t)]\partial_x v_a(t,x)+f_F[S(t)](x),\quad (t,x)\in[0,T]\times\R,\\
\partial _t v_b(t,x)&=p^{B-A}[S(t)]\partial_xv_b(t,x)+f_C[S(t)](x),\quad (t,x)\in[0,T]\times\R.
\end{aligned}
\end{equation*}
\end{thm}

In order to guarantee that the bid and ask price, the spread, and the volume density functions are non-negative, the following assumption has to be satisfied.

\begin{ass}\label{pos}
\begin{eqnarray*}
p^B(s)=0,\quad p^D(s)&=&0\quad\forall\ s\in\left\{(A,v_a,B,v_b)\in E\times E:\ A=B\right\},\\
p^A(s)&=&0\quad\forall\ s\in\left\{(A,v_a,B,v_b)\in E\times E:\ B=0\right\},\\
\p\left(\omega_k^{(n)}<v_{b,k}^{(n)}\left(\pi_k^{(n)}\right),\phi_k^{(n)}=C\right)&=&0\quad \forall\ n\in\N_0,\ k=0,1,\dots,\lfloor T/\Delta t^{(n)}\rfloor-1,\\
\p\left(\omega_k^{(n)}<v_{a,k}^{(n)}\left(\pi_k^{(n)}\right),\phi_k^{(n)}=F\right)&=&0\quad \forall\ n\in\N_0,\ k=0,1,\dots, \lfloor T/\Delta t^{(n)}\rfloor-1.
\end{eqnarray*}
\end{ass}

\begin{lem}
If in addition Assumption \ref{pos} is satisfied in the statement of Theorem \ref{limit2}, then both volume density functions, both prices, and the spread between the best ask price and the best bid price are non-negative.
\end{lem}

\subs{Simulations}\label{sim}

This subsection contains a simulation study of the two-sided limit order book model introduced in the previous subsection. It provides an example of the usefulness of the general dependence structure that is allowed in our model, but was not covered in the previous work \cite{HP}.

For this let us fix $h>0$. We define the imbalance factor in the $n$-th model at time $t_k^{(n)}$ via
\[Im_k^{(n)}:=\frac{VolBid^{(n)}_k}{VolAsk^{(n)}_k+VolBid^{(n)}_k}\]
with 
\[VolAsk_k^{(n)}:=\int_{0}^{h}v_{a,k}^{(n)}(x)dx,\qquad VolBid_k^{(n)}:=\int_{-h}^{0}v_{b,k}^{(n)}(x)dx.\]
Moreover, in the $n$-th model the spread at time $t_k^{(n)}$ is defined as
\[Sp_k^{(n)}:=A_k^{(n)}-B^{(n)}_k.\]

There is strong empirical evidence, cf.~e.g.~\cite{Cao} and \cite{YangZhu}, that the probability whether the next price change is upwards or downwards depends on the imbalance of the bid and ask queues at the top of the book: if the imbalance is high, i.e.~the standing volume at the top of the buy side is significantly higher than the standing volumes at the top of the sell side, then prices are more likely to move upwards. And conversely, an imbalance factor which is close to zero will increase the probability of the price moving downwards. Thus, we choose
\begin{eqnarray*}
\p\left(\left.\phi_k^{(n)}=A\ \right|S^{(n)}_j,\ j\leq k\right)&=&\Delta p^{(n)}\left(1-Im_{k}^{(n)}\right)\cdot\exp\left(-Sp_{k}^{(n)}\right),\\
\p\left(\left.\phi_k^{(n)}=D\ \right|S^{(n)}_j,\ j\leq k\right)&=&\Delta p^{(n)}Im_{k}^{(n)}\cdot\exp\left(-Sp_{k}^{(n)}\right),\\
\p\left(\left.\phi_k^{(n)}=B\ \right|S^{(n)}_j,\ j\leq k\right)&=&\Delta p^{(n)}Im_{k}^{(n)}\cdot\left(1-\exp\left(Sp_k^{(n)}\right)\right), \\
\p\left(\left.\phi_k^{(n)}=E\ \right|S^{(n)}_j,\ j\leq k\right)&=&\Delta p^{(n)}\left(1-Im_k^{(n)}\right)\cdot\left(1-\exp\left(-Sp_{k}^{(n)}\right)\right).
\end{eqnarray*}
Moreover, for simplicity we suppose that placements outside the spread are always of size $50$ (this assumption can easily be relaxed) and are located relatively to the best ask resp.~bid price according to a Gaussian distribution, i.e.
\begin{eqnarray*}
\p\left(\left.\phi_k^{(n)}=C,\ \omega_k^{(n)}=50,\ \pi_k^{(n)}\in dy\ \right|S_{j}^{(n)},\ j\leq k\right)&=&\frac{1}{\sqrt[4]{32\pi^2}}\left(1-\Delta p^{(n)}\right)\left(1-Im_{k}^{(n)}\right)\exp\left(-y^2\right)dy,\\
\p\left(\left.\phi_k^{(n)}=F,\ \omega_k^{(n)}=50,\ \pi_k^{(n)}\in dy\ \right|S_{j}^{(n)},\ j\leq k\right)&=&\frac{1}{\sqrt[4]{32\pi^2}}\left(1-\Delta p^{(n)}\right)Im_{k}^{(n)}\cdot\exp\left(-y^2\right)dy.
\end{eqnarray*}
We note that if integrated over the whole real axis the above two terms sum up to $\left(1-\Delta p^{(n)}\right)/2$. 
Furthermore, cancelations are supposed to be proportional to the current volume. Hence, we choose for all $x\leq 0$,
\begin{align*}
\p\left(\left.\phi_k^{(n)}=C,\ \omega_k^{(n)}\in dx,\ \pi_k^{(n)}\in dy\ \right|S_{j}^{(n)},\ j\leq k\right)=\qquad\qquad\qquad\qquad\qquad\qquad\qquad\qquad\qquad\\
\frac{1}{v^{(n)}_{b,k}(y)}\1_{\left[-v_{b,k}^{(n)}(y),0\right]}(x)\frac{1}{\sqrt[4]{32\pi^2}}\left(1-\Delta p^{(n)}\right)Im_{k}^{(n)}\cdot\exp\left(-y^2\right)dxdy\\
\p\left(\left.\phi_k^{(n)}=F,\ \omega_k^{(n)}\in dx,\ \pi_k^{(n)}\in dy\ \right|S_{j}^{(n)},\ j\leq k\right)=\qquad\qquad\qquad\qquad\qquad\qquad\qquad\qquad\qquad\\
\frac{1}{v^{(n)}_{a,k}(y)}\1_{\left[-v_{a,k}^{(n)}(y),0\right]}(x)\frac{1}{\sqrt[4]{32\pi^2}}\left(1-\Delta p^{(n)}\right)\left(1-Im_{k}^{(n)}\right)\exp\left(-y^2\right)dxdy.
\end{align*}
If we integrate the above two terms in both variables and sum them, we get again $\left(1-\Delta p^{(n)}\right)/2$. 

For the simulation we choose the following parameter values: 
\begin{eqnarray*}
n=250,\qquad \alpha=0.8,\qquad h=0.55, \qquad T= 8.
\end{eqnarray*}

We start with a limit order book that has a severe imbalance at time $t=0$: standing volumes at the ask side are many times higher than standing volumes at the bid side. However, our specifications about the conditional distributions of incoming orders made above guarantee that the imbalance at the top of the book will disappear after some time.
While Figure \ref{3D} shows the evolution of the whole visible limit order book plotted at different points in time, Figure \ref{2D} only shows the evolution of the bid and ask price together with the cumulated volumes at the top of the book.

\begin{figure}[h]\label{book}
\includegraphics[width=0.8\textwidth]{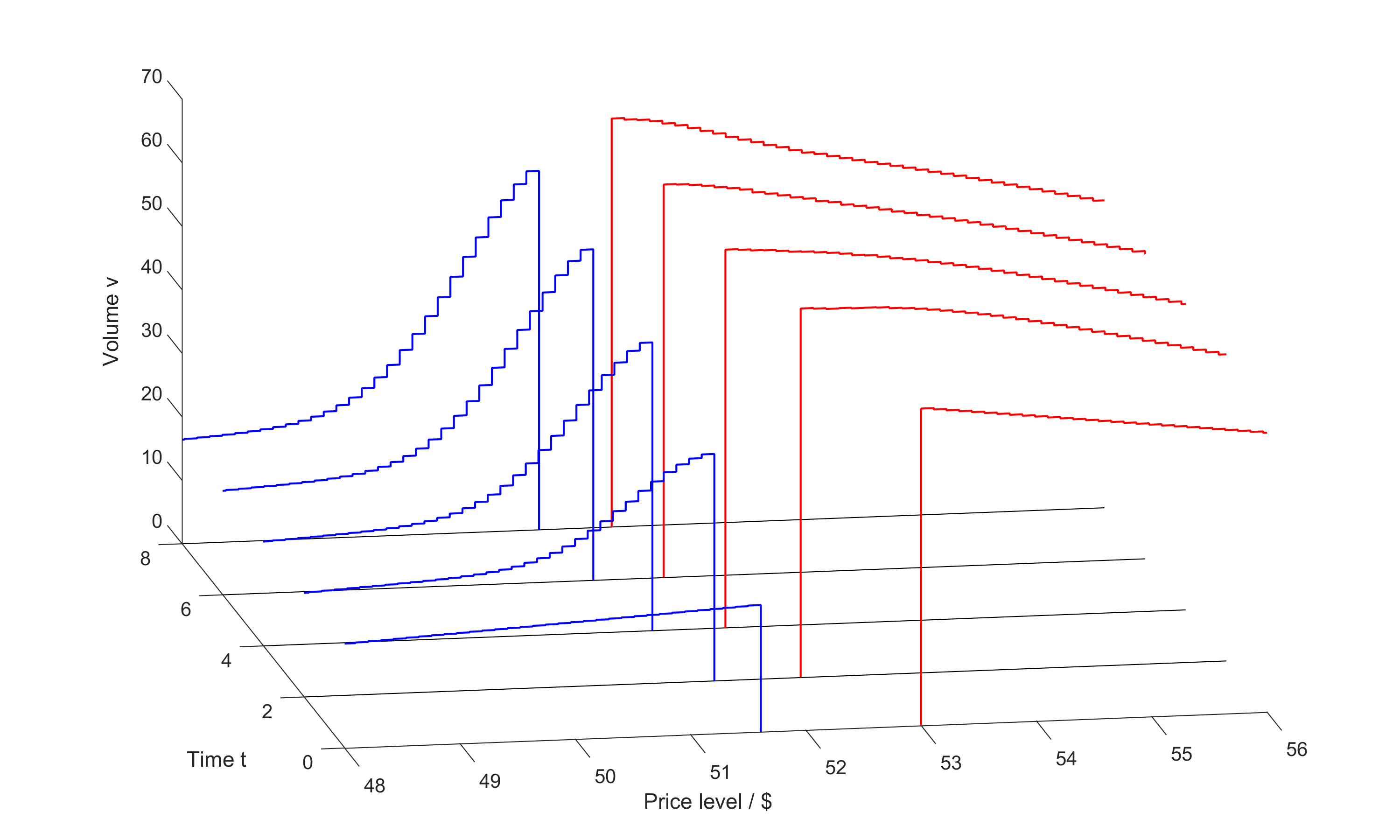}
\caption{The evolution of the limit order book volumes}
\label{3D}
\end{figure}

\begin{figure}[h]\label{graphs}
\includegraphics[width=0.8\textwidth]{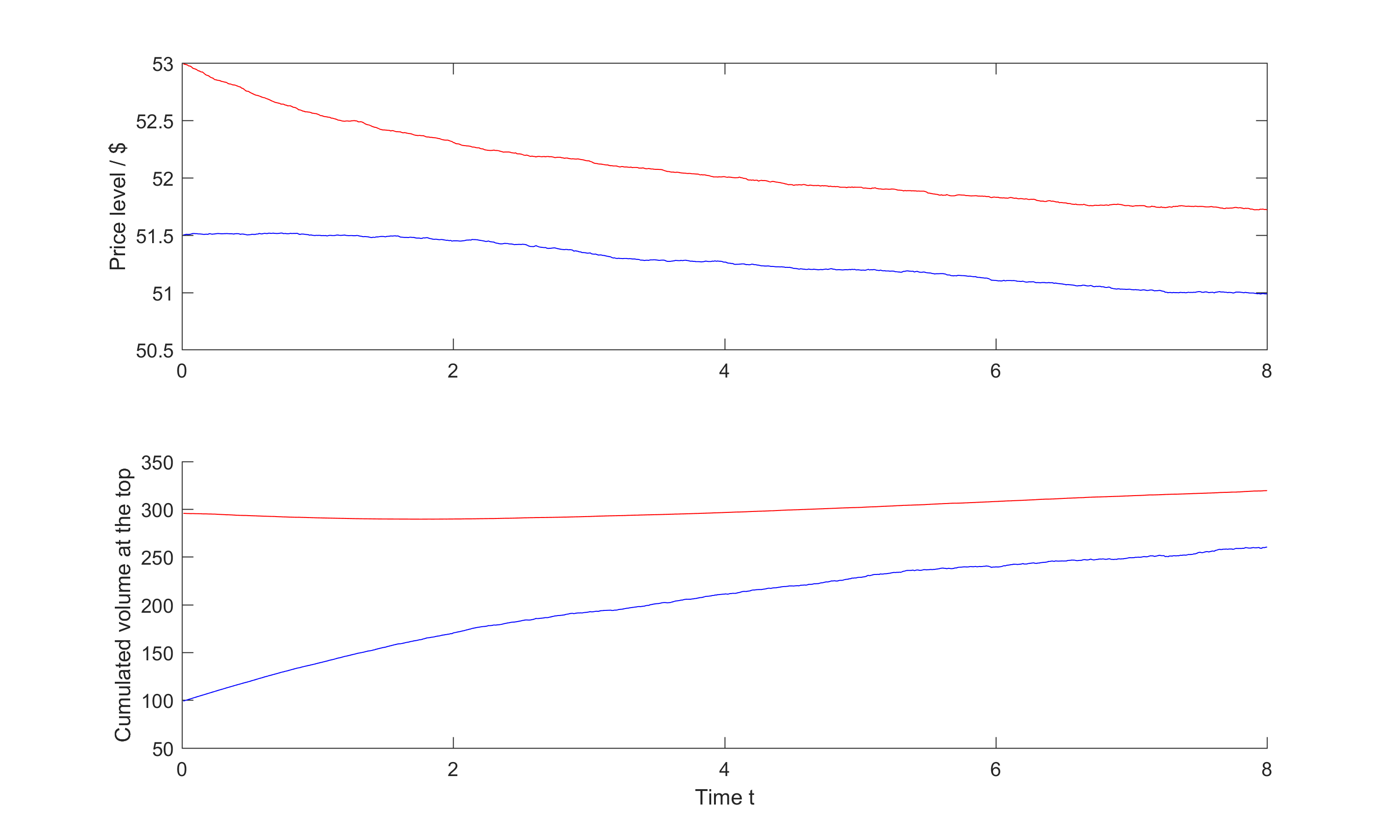}
\caption{Evolution of prices and of volumes at the top of the book}
\label{2D}
\end{figure}

The evolution of prices is influenced by the spread as well as the imbalance factor. As can be seen from the first graph in Figure \ref{2D}, first the spread is the dominating factor (due to the exponential) and forces the ask price to decrease faster than the bid price. However, after time $t\approx3$ both prices decrease simultaneously keeping the spread almost constant. This downward shift of the mid price is caused by the imbalance of standing volumes at the top of the book and therefore mimicks the findings in the literature very well, cf.~e.g.~\cite{Cao}. Also the second graph in Figure \ref{2D} shows that cumulated volumes at the top of both sides of the limit order book converge. In this particular simulation study the buy side volume approaches the sell side volume because we have chosen the size of order placements much greater than the size of average cancelations for the initial volume density functions. If placements were supposed to be of a much smaller size, for example 10 instead of 50, the opposite effect could be observed, i.e.~the sell side volumes at the top of the book would decrease to approach the buy side volumes at the top of the book. Of course, much more general random and even state dependent choices of the order sizes are possible and would lead to even more interesting dynamics. 

Last but not least, Figure \ref{3D} shows that the discrete order book dynamics can indeed be well approximated by a smooth function and is hence supportive of Theorem \ref{limit2}.

\appendix

\s{}

\subs{A weak law of large numbers for triangular martingale difference arrays}

The following weak law of large numbers for triangular martingale difference arrays relies on moment estimates from \cite{Pisier} and is shown in \cite{HP}. 

\begin{thm} \label{WLLN}
Let $\left(y_k^n, k=1,\dots, n; n\in\N\right)$ be a triangular martingale difference array taking values in a real separable Hilbert space such that
\[\sup_{\substack{k\leq n\\n\in\N}}\left(n^{2\alpha}\E|y_k^n|^2\right)<\infty\]
for some $\alpha>\frac{1}{2}$. Then for all $\eps>0$,
\[\lim_{n\ra\infty}\p\left(\sup_{m\leq n}\left|\sum_{k=1}^m y_k^n\right|>\eps\right)=0.\]
\end{thm}

\subs{Gronwall lemmas}

For reference we cite the discrete and continuous version of Gronwall's lemma which can for example be found in \cite{Elyadi}:

\begin{lem}\label{Gronwall}
Let $(y_m)_{m\geq0}$, $(f_m)_{m\geq0}$, and $(g_m)_{m\geq0}$ be nonnegative sequences. If
\[y_m\leq f_m+\sum_{k=0}^{m-1}g_ky_k\quad \forall\ m,\]
then
\[y_m\leq f_m+\sum_{k=0}^{m-1}f_kg_ke^{\sum_{j=k+1}^{m-1}g_j}.\]
\end{lem}

\begin{lem}\label{Gronwall2}
Let $\alpha, \beta, u$ be real-valued functions defined on some interval $I = [a, b],\ a < b,\ a,b,\in\R$.  Assume that $\beta$ and $u$ are continuous and that the negative part of $\alpha$ is integrable on every compact subinterval of $I$.
If $\beta$ is non-negative and if $u$ satisfies the integral inequality
\[u(t) \leq \alpha(t) + \int_a^t \beta(s) u(s)ds\qquad \forall\ t\in I,\]
then
\[u(t) \leq \alpha(t) + \int_a^t\alpha(s)\beta(s)\exp\left(\int_s^t\beta(r)dr\right)ds,\qquad t\in I.\]
\end{lem}

{\small

\bibliographystyle{abbrv}

\bibliography{LOBLit}
}

\end{document}